\documentclass[journal]{IEEEtran}
\usepackage{graphics}
\usepackage{graphicx}
\usepackage{epsfig}
\usepackage{amsmath,amssymb}
\usepackage{mathrsfs,amsfonts,fancyhdr}
\usepackage{multirow}
\usepackage{multicol}
\usepackage{cite}
\usepackage{amsthm}
\newtheorem{thm}{Theorem}[]

\theoremstyle{remark}
\newtheorem{rem}{Remark}[]
\begin{document}
\title{Decode-Forward Transmission for the Two-Way Relay Channels}
\author{\normalsize{Ahmad Abu Al Haija, \emph{Student Member, IEEE}, Peng Zhong  and Mai Vu, \emph{Senior Member, IEEE}}
\thanks{ A part of this work was presented at
  the IEEE International Conference on Communications (ICC) $2012$ and
  IEEE Information Theory Workshop (ITW) $2011$. This work was supported in part by grants from
 the Natural Science and Engineering Research Council of Canada (NSERC),
 the Fonds Quebecois de la
Recherche sur la Nature et les Technologies (FQRNT) and
the Office of Naval Research (ONR, Grant N00014-14-1-0645). Any opinions, findings, and conclusions or recommendations expressed in this material are those of the authors and do not necessarily reflect the views of the Office of Naval Research.}
\thanks{Ahmad Abu Al Haija is with the Department of
  Electrical and Computer Engineering, McGill University, Montreal,
  Canada (e-mail: ahmad.abualhaija@mail.mcgill.ca). This work was performed while he visited Tufts University. Peng Zhong was with the ECE department at McGill University when this work was performed.  Mai Vu is with the
  department of Electrical and Computer Engineering at Tufts
  University (email: maivu@ece.tufts.edu).
}}

\maketitle
\vspace{-20mm}

\begin{abstract}
We propose composite  decode-forward (DF) schemes for the two-way relay channel in both the full- and half-duplex modes by combining coherent relaying, independent relaying and partial relaying strategies.  For the full-duplex mode, the relay partially decodes each user's information in each block and forwards this partial information coherently with the source user to the destination user in the next block as in block Markov coding. In addition, the relay independently broadcasts a binning index of both users' decoded information parts in the next block as in independent network coding. Each technique has a different impact on the relay power usage and the rate region.  We further analyze in detail the independent partial DF scheme and derive in closed-form link regimes  when this scheme achieves a strictly larger rate region than just time-sharing between its constituent techniques, direct transmission and independent DF  relaying, and when it reduces to a simpler scheme.   For the half-duplex mode, we
propose a $6$-phase time-division scheme that incorporates all considered relaying techniques and uses joint decoding simultaneously over all receiving phases. Numerical results show
significant rate gains over existing DF schemes, obtained by performing link adaptation of the composite scheme based on the identified link regimes.
%
\end{abstract}
%

\section{Introduction}
The two-way relay channel (TWRC) has received attention as
one of the basic cooperation models in both wired and wireless
networks. The existence of different routes in a wired network or the
shared-medium nature of a wireless network allows a relay node to help
exchange information between two source nodes. The application is apparent
in infrastructure-less networks such as ad hoc and sensor networks, where
two nodes can communicate with help from a third node in between. In infrastructure-aided networks such as the cellular system, the
newly introduced device-to-device communication mode \cite{DTD} also allows
the application of two-way relay transmission, where the relay could be an idle user or a
nearby femto base station helping the direct communication between two active users. Such a
scenario is appealing in high speed video exchange for example.

Multiple two-way relaying strategies have been proposed for both
full- and half-duplex communications \cite{rankov2006achievable,
  xie2007network, PV_DC, kim2008per, zhong2012partial, ASHARA, nw6p,
  khafagyicc, gerdes}, each strategy
has slightly different rate region  performance from the others. On the one hand, it is of theoretical interest
to find a composite scheme that achieves the highest performance and
includes all known schemes. On the other hand, it is of practical
interest to design the simplest scheme that achieves the maximum
performance. These seemingly contradicting interests require a
thorough understanding of the effect  on performance of each separate relaying
strategy and the condition for this effect to be
maximized. Such a condition may appear in the form of regimes on the network links
 strength. The more complicated scheme is then only
needed under specific link regimes. Mapping out these link regimes is
related to the optimal resource allocation to the constituent
techniques in a composite scheme, but the question here is not to find
the exact optimal resource allocation to all techniques, but to ask
when it is optimal to allocate zero resources to a certain technique
(so that the optimal scheme need not use that technique). This topic
has not been explored much in the literature.
%
%
\subsection{Related works}
Several full-duplex schemes have been proposed for the TWRC,
 differing mainly on the relaying techniques employed at the
 relay. Inspired by the basic techniques for the classical one-way
 relay channel \cite{cover1979capacity}, different decode-and-forward
 and compress-and-forward relaying schemes have been applied to the
 TWRC. In this paper, we focus on the decode-and-forward (DF)
 strategies where the relay fully or partially decodes the information
 of each user and forwards it to the other user in the next
 transmission block. A DF strategy based on block Markov encoding is
 proposed in \cite{rankov2006achievable}, where in each transmission
 block, both users send their new information and also repeat the old
 information of the previous block. Different from the regular one-way relay channel, here each user's old information is sent coherently with only a part of the relay's transmit signal, but not with the whole relay signal, which contains information of both users and is unknown at either user. In another DF strategy
 without block Markovity \cite{xie2007network}, the users do not
 repeat their old messages and the relay forwards only an independent
 codeword as a function of both users' messages without coherent
 transmission with either user. This independent coding strategy
 resembles the idea of network coding at the relay. Each of these two
 strategies has a different impact on the achievable rate and the
 power allocation at the relay; neither strategy always outperforms the
 other one.

 Partial DF (pDF) relaying has not been investigated much in the
 literature, except when it is in combination with another relaying
 strategy such as compress-and-forward \cite{CST1}. This is most
 likely because of the belief that partial DF alone does not bring
 improvement over full DF relaying as in the classical single-antenna one-way  relay channel
 \cite{gmze}. Rather surprisingly, we showed   that for the two-way relay channel, partial DF can
 strictly improve performance over full DF \cite{zhong2012partial}. The improvement comes from
 the fact that signals from multiple users can now interfere with each
 other, hence in certain link regimes, it is better to reduce this
 interference by performing decode-and-forward partially.

 For the half-duplex mode, existing works differ mainly in the number
 of phases, the relaying technique at the relay, and the decoding
 technique at the users. A $4$-phase independent DF is proposed in
 \cite{kim2008per} where the two users alternatively transmits and
 receives during the first two phases, both users transmit in phase
 $3$ and the relay transmits in phase $4$. The most general strategy
 for half-duplex TWRC includes $6$ phases, which has been considered
 in several works including \cite{ASHARA, nw6p}, where in each of the
 additional $2$ phases, one user and the relay independently transmit
 to the other user. The degrees of freedom is analyzed for a $6$-phase
 independent DF scheme with multiple-antenna nodes in
 \cite{khafagyicc}. Another $6$-phase scheme considering both
 independent and coherent pDF relaying is proposed in \cite{gerdes},
 where coherent transmissions occur in phases $5$ and $6$. However,
 the users perform separate decoding at the end of each phase instead
 of simultaneous  decoding over all phases.

Although many schemes have been proposed for the TWRC, few works
analyzed them. Most existing analysis is about the optimal resource
allocation for a simple scheme. For example, a $2$-phase half-duplex
DF scheme without the direct links between users is optimized for the
phase durations to maximize the rate region \cite{chenicc} and the
sum rate \cite{op2p}. An OFDMA $2$-phase DF scheme is optimized for
the resource allocations that achieve the largest achievable rate region per
OFDMA subcarrier \cite{ROFD1} or jointly over all multi-carriers
\cite{ROFD2}. Analysis for more complex schemes which are composed of
several techniques is not yet seen in the literature.
%
\subsection{Main results and contributions}
In this paper, we investigate a composite scheme for the TWRC based on
decode-and-forward relaying that includes several techniques: coherent
block Markov relaying, independent relaying, and partial relaying. Such a combination of DF techniques has not been considered previously in the literature. We
design and analyze the composite schemes in both the full- and half-duplex modes. The
full-duplex design allows a detailed analysis on the performance impact of
the considered techniques under different link regimes, whereas the
half-duplex adaptation allows application in current wireless
systems.

In our proposed schemes, each user splits its message into a common
and a private part, the common message is to be decoded at the relay,
whereas the private message is to be decoded only at the destination. For
full-duplex transmission, each user superposes its private
part over the common part of the current block and the common part of
the previous block. The relay decodes the current common messages of
both users and forward each previous common message coherently with
its source user to the destination user. In addition, the relay encodes a
bin index that is a function of both users' common messages and forwards it
independently with the signals from the users. The relay's signal therefore is a superposition of three parts.  Each user employs sliding
window decoding over two consecutive blocks to decode the other user's
messages. This scheme is a generalization of our previous schemes in
conference publications \cite{PV_DC, zhong2012partial}; it also includes
all existing schemes in \cite{rankov2006achievable, xie2007network} as
special cases and achieves a strictly larger rate region than the
time-shared region between these schemes.

We then analyze in detail the independent coding and partial DF relaying components of our
scheme to understand when each of these constituent
techniques is necessary. Without one or both these techniques, the scheme
reduces to two-way full DF relaying, direct transmission or a hybrid one-way full DF relaying for one user and direct transmission for the other user. We derive conditions on link strengths for when it is sufficient to use one of these  simplified schemes to achieve the maximum achievable rate region and when it is necessary to perform the composite independent partial DF. Such a link-regime identification has not been done before, and provides the necessary insights for applying two-way relaying in practice.
These link regimes also reveal that independent partial relaying is beneficial in the TWRC when the relay link from one user is slightly stronger
than its direct link, quantified by conditions among the received SNRs over these links. Because of space and complexity constraints, link-regime  analysis is performed separately for the independent
and coherent full DF relaying components in \cite{lisa2015}.
%
%

For the half-duplex mode, we propose an inclusive $6$-phase
time-division scheme. Different from existing $6$-phase schemes
\cite{ASHARA, nw6p, khafagyicc, gerdes}, our scheme utilize all coherent,
independent and partial relaying techniques; furthermore, both users employ
simultaneous decoding over all phases instead of separate decoding in each
phase. Hence, our scheme outperforms all existing schemes.

The rest of this paper is organized as follows. Section \ref{sec: CM}
describes the TWRC model for both the full- and half-duplex
modes. Section \ref{fd-scheme} describes the proposed composite pDF
scheme for the full-duplex mode and derives its achievable rate
region. Section \ref{sec: AFDS} analyzes the link regimes for the
independent pDF scheme. Section \ref{sec: HDS} describes an inclusive
$6$-phase pDF scheme for the half-duplex mode. Section \ref{sec: NR}
presents numerical results and Section \ref{sec: CN} concludes the
paper.
\section{Two-way relay channel model}\label{sec: CM}

The two-way relay channel consists of two user nodes, each with its
own information to send to the other node, and a relay node that can
receive from and transmit to both user nodes to facilitate
communication between them.
We are interested in the full TWRC with a direct link between the two
user nodes in addition to the links to the relay which captures the full
interaction in a wireless environment.

Next we will present both full and half-duplex channel
models as shown in Figure \ref{fig:6_phases}. Traditionally, wireless
communication devices are half-duplex, meaning that they can only
transmit or only receive on a single frequency band at a time. Recently,
full-duplex wireless has been demonstrated to work even with just a single
transmit and a single receive antenna, bringing the possibility of
functional full-duplex wireless systems in the future. In this paper, we
will analyze both the full and half-duplex modes for the
TWRC. Often the design and properties of a transmission scheme for the
full-duplex mode can be applied to the half-duplex mode. Hence considering
both modes helps us focus on the design  and analysis of transmission schemes for the
full-duplex mode first, then extend these schemes to satisfy the
half-duplex constraint.
\begin{figure}[!t]
    \begin{center}
    \includegraphics[width=0.95\textwidth]{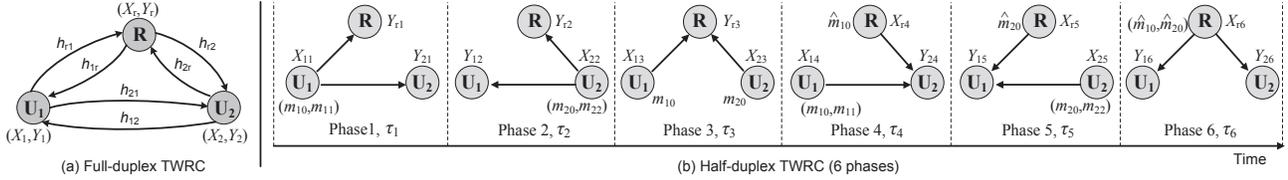}
    \caption{ Full and Half-duplex model for TWRC.}
    \label{fig:6_phases}
    \end{center}
    \vspace{-6mm}
\end{figure}
\subsection{Full-duplex channel model}

The full-duplex TWRC as shown in Figure \ref{fig:6_phases}(a) can be modeled as
\begin{align}
  Y_1&=h_{12}X_2+h_{1r}X_r+Z_1, \nonumber \\
  Y_2&=h_{21}X_1+h_{2r}X_r+Z_2, \nonumber \\
  Y_r&=h_{r1}X_1+h_{r2}X_2+Z_r,
\label{GTWRC}
\end{align}
where $Z_1, Z_2, Z_r\sim\mathcal{CN}(0,1)$ are independent complex Gaussian
noises with zero mean and unit variance, and $(X_1,Y_1)$, $(X_2,Y_2)$,
$(X_r,Y_r)$ are pairs of the transmitted signal and received signal at user
$1$ $(\mathbf{U}_1)$, user $2$ $(\mathbf{U}_2)$ and the relay,
respectively. The average input power
constraints for $\mathbf{U}_1$, $\mathbf{U}_2$ and the relay are assumed
to be $P_1,$ $P_2$ and $P_r$, respectively.

The channel gain coefficients $h_{12}, h_{1r}, h_{21}, h_{2r}, h_{r1},
h_{r2}$ are complex value. We use the standard assumption in coherent relaying literature \cite{cover1979capacity, gmze, gamal2010lecture}
  that the phases of these
channel coefficients are known at the respective transmitters so that
coherent transmission from two different transmitters are possible,
and the full channel coefficients are known at the respective
receivers. As such, the achievable rate will depend only on the
{\it amplitude} of the channel coefficients, which we will denote as
$g$ with the same subscripts as $h$.
\subsection{Half-duplex channel model}
Time division is used to adapt a full-duplex transmission scheme to
the half-duplex mode. The main idea is that each transmission block,
(or equivalently a time slot or a frame) is divided into several
phases. Depending on the number of phases, different half-duplex
protocols can result. Existing literature has a number of variation with
$2$, $3$, $4$ or $6$ phases \cite{kim2008per, zhong2012partial, ASHARA,
  nw6p, khafagyicc, gerdes}.

In this paper, we will consider the most comprehensive and inclusive
$6$-phase protocol. The $6$-phase protocol encompasses all other
half-duplex protocols as special cases, which can be obtained by setting
the duration of some of the $6$ phases to zero. The channel model for each
phase, as depicted in Figure \ref{fig:6_phases}(b), can be expressed as
\begin{align}
 \!\!\!\!\text{Phase 1:} \; Y_{21}&=h_{21}X_{11}+Z_{21},\;
   Y_{r1}=h_{r1}X_{11}+Z_{r1},\nonumber\\
  \!\!\!\!\text{Phase 2:} \; Y_{12}&=h_{12}X_{22}+Z_{12},\;Y_{r2}=h_{r2}X_{22}+Z_{r2},\nonumber
\end{align}
\\ \\ \\ \\ \\ \\ \\
\begin{align}
  \label{6_signaling}
  \!\!\!\!\text{Phase 3:} \; Y_{r3}&=h_{r1}X_{13}+h_{r2}X_{23}+Z_{r3},\nonumber\\
  \!\!\!\!\text{Phase 4:} \; Y_{24}&=h_{21}X_{14}+h_{2r}X_{r4}+Z_{24},\nonumber  \\
  \!\!\!\!\text{Phase 5:} \; Y_{15}&=h_{12}X_{25}+h_{1r}X_{r5}+Z_{15},\nonumber\\
  \!\!\!\!\text{Phase 6:} \; Y_{16}&=h_{1r}X_{r6}+Z_{16},\; Y_{26}=h_{2r}X_{r6}+Z_{26},
\end{align}
where $X_{ij}$ and $Y_{ij}$, ($i \in [1,2,r], j \in [1,\ldots,6]$)
respectively represents the transmitted and received signal of user
$i$ during phase $j$. All the noises $Z_{ij}$ are i.i.d $\sim\mathcal{CN}(0,1).$  The channel
coefficients are similarly defined as in the full-duplex model.
\section{A full-duplex decode-forward transmission scheme}
\label{fd-scheme}
The unique feature in
designing a DF scheme for the TWRC arises from the fact that each user does not have
the complete information decoded at the relay, as was the case in the
one-way relay channel. There are two different ways to build a DF
scheme for the TWRC, depending on whether the source performs block Markov encoding on only a part of the relay's signal,
or performs coding independently with the relay. Further, partial DF has not been previously considered for the TWRC but
can strictly outperform full DF, as shown later in this paper.
Next, we propose a new composite pDF scheme for the TWRC  that combines all three strategies: partial DF, coherent block
Markov coding and independent coding.
\vspace{-4mm}
\subsection{Transmission scheme}
Each user node splits its message
into two parts: a common message and a private message. The user first
encodes the common message by superposing this message on the common
message of the previous block as in block Markov coding;
it then superposes the private message on the current common
message. The relay only decodes the current common message of each user; it
then performs random binning of the pair of decoded common messages and
generates a codeword for this bin index, superposes this codeword on the
two codewords for two common messages, effectively combining relay
strategies of both \cite{xie2007network} and
\cite{rankov2006achievable}. The relay forwards this codeword containing
the two decoded common messages and their bin index in the next block. Each
user decodes the other user's message by sliding window decoding based on
received signals in both the current and previous blocks.

\subsubsection{Transmit signal design}
Denote the new messages user 1 and user 2 send in block $i$ as
$m_{1,i}$ and $m_{2,i}$, respectively. Each user splits its own
message as $m_{1,i} = (m_{10,i},m_{11,i})$ and $m_{2,i} =
(m_{20,i},m_{22,i})$, where the first part is the common message and
the second is the private message. The relay partitions the set
of all common messages $\{m_{10,i-1},m_{20,i-1}\}$ equally and
uniformly into a number of bins and indexes these bins by
$\{l_{i-1}\}$. The users and the relay then construct the transmit
signals in block $i$ as follows:
\begin{align}
  X_1&=\sqrt{\alpha_1}W_1(m_{10,i-1})+\sqrt{\beta_1}U_1(m_{10,i})
  +\sqrt{\gamma_1}Q_1(m_{11,i}) \nonumber\\
  X_2&=\sqrt{\alpha_2}W_2(m_{20,i-1})+\sqrt{\beta_2}U_2(m_{20,i})
  +\sqrt{\gamma_2}Q_2(m_{22,i}) \nonumber\\
 X_r&=\sqrt{k_1 \alpha_1}W_1(m_{10,i-1})+
  \sqrt{k_2 \alpha_2}W_2(m_{20,i-1}) \nonumber\\
  &\;\;+\sqrt{\beta_3}U_r(l_{i-1})
  \label{tx-signals}
\end{align}
where $W_1, W_2, U_1, U_2, U_r, Q_1, Q_2$ are independent Gaussian
signals with zero mean and unit variance that encode the respective
messages and bin index. Power allocation factors $\alpha_1,
\alpha_2, \beta_1, \beta_2, \beta_3,$ $\gamma_1, \gamma_2$ are
non-negative and satisfy
\begin{align}
   \alpha_1+\beta_1 +\gamma_1&\leq P_1, \nonumber\\
  \alpha_2+\beta_2 +\gamma_2 &\leq P_2, \nonumber \\
  k_1 \alpha_1+k_2 \alpha_2+\beta_3 &\leq P_r,
  \label{power-constraint}
\end{align}
where $k_1, k_2$ are scaling factors that relate the power allocated
to transmit the same common message at each user and the relay.

We note several features of the signal construction in
\eqref{tx-signals}. In each transmission block, each user sends
not only new messages of that block (with index $i$), but also repeats
the common message of the previous block (with index $i-1$).  The relay
always transmits information of the previous block which it decodes
only at the end of the previous block. Thus the block Markov structure creates a
coherency between the common signal parts transmitted from each
user and the relay (denoted by $W_1$ and $W_2$) which results in a
beamforming gain.

The relay not
only forwards the signals $W_1$ and $W_2$
but also creates a new signal $U_r$ that independently
encodes both common messages via binning. The random binning is a function of both messages
decoded at the relay and resembles network coding.

Each of these two relaying strategies has a unique implication. Markov
block coding brings a coherent beamforming gain between a source and
the relay, but requires the relay to split its power between $W_1$ and
$W_2$ for each pair of common messages. Independent network coding via
binning, on the other hand, let the relay use its whole power to send
the bin index of the message pair. This bin index can then solely represent
one common message when decoding at each user because of side information
on the other common message available at this user, as can be seen in the
decoding discussion below.

\subsubsection{Decoding rules}
Decoding occurs at all nodes, including the relay and the users. We
assume all messages are equally likely and all nodes use maximum
likelihood (ML) sequence detection, which in this case is equivalent
to the optimal decoding rule presented next.\footnote{In deriving the
  rate region in the Appendix, however, we use joint typicality
  decoding, which is simpler to analyze than ML decoding and achieves
  the same rate region.}

At the relay, decoding is quite simple and is similar to decoding in the
multiple access channel. The received signal in block $i$ at the relay is
\begin{align}
  Y_r&=h_{r1} \left( \sqrt{\alpha_1}W_1+\sqrt{\beta_1}U_1
    +\sqrt{\gamma_1}Q_1 \right)\\
  &\;\;+ h_{r2}\left( \sqrt{\alpha_2}W_2+\sqrt{\beta_2}U_2
    +\sqrt{\gamma_2}Q_2 \right) + Z_r.\nonumber
\end{align}
In block $i$, the relay already knows $W_1,W_2$ (which carry
$m_{10,i-1},m_{20,i-1}$) and is interested in decoding ${U}_1$ and
${U}_2$ (which carry $m_{10,i},m_{20,i}$). For the purpose of
highlighting what is known, we write the optimal maximum a posteriori
probability (MAP) decoding rule, which can be converted to ML decoding
rule, as
\begin{align}
  \left( m_{10,i}, m_{20,i} \right) = \arg \max P(U_{1,i},U_{2,i}|Y_{r,i},W_{1,i},W_{2,i}).
\end{align}
In this decoding, the relay treats signals ($Q_1,Q_2$) that carry the private information as
noise.

Decoding at each user node is based on signals received in two
consecutive blocks. To decode new information sent in block $i$, a
user will use received signals in both blocks $i$ and $i+1$, which
results in a one-block decoding delay. Take user 2 for example, this user's
received signal in each block is
\begin{align}
 Y_2 &= h_{21} \left( \sqrt{\alpha_1}W_1+\sqrt{\beta_1}U_1
    +\sqrt{\gamma_1}Q_1 \right)\\
  &\;\;+ h_{2r} \left( \sqrt{k_1 \alpha_1}W_1+
    \sqrt{k_2 \alpha_2}W_2 +\sqrt{\beta_3}U_r\right)
  + Z_2.\nonumber
\end{align}
User 2 knows its own signal $W_2$ and can directly subtract it from the
received signal.

At the end block $i$, each user has presumably decoded all previous blocks
information correctly. That means in block $i$, user 2 knows $W_1,U_r$
(which encode $w_{10,i-1}$) and thus only needs to decode $U_1, Q_1$
(which encode $w_{10,i}, w_{11,i}$). In block $i+1$, the relevant signals
carrying the same wanted information are $W_1$ and $U_r$ (which encode
$w_{10,i}$). Note that even though $U_r$ encodes the bin index of
the pair of common messages and hence carries information about both
common messages, but because one common message is originated and thus is
known at each user, $U_r$ can be used at each user solely to decode the other
common message. This is the effect of side information in
decoding at each user in the TWRC, which has also been noted in
\cite{kim2008per}. For block $i+1$, user 2 will decode $W_1$ and $U_r$ while
treating the other signals $U_1, Q_1$ as noise.

We write this joint decoding over two blocks using MAP decoding at user 2 as follows:
\begin{align}
 \left( m_{10,i}, m_{11,i} \right) = \arg \max \big(
  P(U_{1,i},Q_{1,i}|Y_{2,i},W_{2,i},W_{1,i},U_{r,i})\nonumber\\ \times
  P(W_{1,i+1},U_{r,i+1}|Y_{2,i+1},W_{2,i+1},m_{20,i}\big).
\end{align}
Note that the above decoding rule applies across two blocks simultaneously, such that the decoded message
pair maximizes {\it the product} of the decoding probabilities in these
blocks.
\vspace{-3mm}
\subsection{Achievable rate region}
With above scheme, we get the
following achievable rate region.
\begin{thm}
\label{thm:complete_pDF_gaussian}
Using the proposed partial decode-forward based scheme, the following
rate region is achievable for the Gaussian two-way relay channel:
\begin{align}
  R_1 &\leq \min\{I_1+I_4, I_{5}\},\nonumber\\
  R_2 &\leq \min\{I_2+I_{6}, I_{7}\},\nonumber\\
  R_1+R_2 &\leq I_3+I_4+I_{6},\;\;\text{where}
  \label{rates-fd}
\end{align}
\begin{align}  \label{rates-constr}
  I_1&=C\left(\frac{g_{r1}^2 \beta_1}
    {g_{r1}^2\gamma_1 +g_{r2}^2\gamma_2 +1}\right), \quad
  I_4=C(g_{21}^2\gamma_1),\\
  I_2&=C\left(\frac{g_{r2}^2 \beta_2 }
    {g_{r1}^2\gamma_1+g_{r2}^2\gamma_2+1}\right), \quad
  I_{6}=C(g_{12}^2\gamma_2),\nonumber\\
  I_3&=C\left(\frac{g_{r1}^2 \beta_1 +g_{r2}^2 \beta_2 }
    {g_{r1}^2\gamma_1+g_{r2}^2\gamma_2+1}\right),\nonumber\\
  I_{5}&=C\left( \alpha_1 \left(g_{21}  + g_{2r} \sqrt{k_1}\right)^2
    + g_{21}^2 (\beta_1 + \gamma_1)  +  g_{2r}^2   \beta_3  \right)\nonumber\\
  I_{7}&=C\left( \alpha_2 \left(g_{12} + g_{1r} \sqrt{k_2}\right)^2
    + g_{12}^2 (\beta_2 + \gamma_2) +  g_{1r}^2   \beta_3  \right), \nonumber
\end{align}
with $C(x)= \log (x+1)$, $g_{*} = |h_{*}|$ as amplitudes of $h_{*}$, and
power factors satisfying \eqref{power-constraint}.
\end{thm}
\begin{IEEEproof}
  To sketch the main ideas, rate constraints $I_1,I_2,I_3$ come directly
  from decoding the common messages at the relay as in a multiple access
  channel. Constraints $I_4$ and $I_6$ come from decoding the private
  messages at each user, assuming the common messages have been decoded
  correctly as at the relay. Constraints $I_5$ and $I_7$ come from the
  joint decoding of both common and private messages at a user. As such,
  the transmission rate of each user (which is the sum of the common and
  private message rates) is constrained by the smaller of the two rates
  achievable by decoding at the other user only, and across the relay and
  the other user. A rigorous proof based on information theoretic analysis
  is available in Appendix A.
\end{IEEEproof}
The expressions $I_5$ and $I_7$ carry several important features that are
direct results of the signal structure at the relay. Recall that the
signal at the relay in \eqref{tx-signals} contains two parts: the block
Markov parts in $W_1,W_2$ and the independent part in $U_r$. The block
Markov parts contribute to a coherent beamforming gain in $I_5$ and
$I_7$. The power allocated to the beamforming signals is separate in
$k_1\alpha_1$ and $k_2\alpha_2$, where each part contributes to the
achievable rate of only one user (either $R_1$ or $R_2$). On the other
hand, the independent signal part has power $\beta_3$ that contributes
to the rates of both users (both $R_1$ and $R_2$). Hence there is an interesting balance in
allocating power to each signal part at the relay as to maximize the
rate region, where beamforming has higher rate but requires a power split
at the relay, and independent signaling has lower rate but can use the
whole power to increase the transmission rate of each user.

By adjusting the power allocation, the proposed scheme encompasses and
improves upon all previous DF-based schemes \cite{rankov2006achievable, xie2007network}.  Specifically, the proposed scheme reduces to the Markov (coherent) DF scheme in \cite{rankov2006achievable} by setting $\gamma_1=\gamma_2=\beta_3=0$; to the independent DF scheme in \cite{xie2007network} by setting $\alpha_1=\alpha_2=\gamma_1=\gamma_2=0$; and to the DT scheme by setting $\alpha_1=\alpha_2=\beta_1=\beta_2=\beta_3=0$.
Moreover, by setting $\gamma_1=\alpha_2=\beta_2=0$, the proposed scheme reduces to the hybrid DF and DT scheme where user $1$ performs coherent and independent DF relaying with the relay while user $2$ performs DT. By further setting $\beta_3=0$ ($\alpha_1=0)$,
 the scheme has user $1$ performs only coherent (independent) DF relaying with the relay. By setting $\gamma_2=\alpha_1=\beta_1=0$,
 we have the opposite of the previous case where user $1$ performs DT and user $2$ performs DF relaying.

As shown later in both analysis and simulation, however, the proposed scheme
performance is not a mere linear combination of the existing schemes
performance. Our scheme can outperform all existing schemes and
their combination by achieving rate points outside the time-sharing region
of existing schemes.
\section{Analysis of the full-duplex scheme}\label{sec: AFDS}
We have proposed a composite DF scheme for the TWRC that utilizes several
techniques: coherent block Markov coding, independent coding and partial
forwarding. Each technique is represented by a signal part in the
proposed transmit signal structure at each node. When a technique is
not needed, its associated signal part will be allocated zero
power. The necessity of each of these techniques depends on the
network configuration, in particular the relative strength of
different links in the network. For some link regimes, the proposed
scheme can reduce to a simpler scheme without using certain techniques
and still achieve the maximum rates, while for other link regimes, all
techniques are necessary in order to achieve the largest rate
region. It is of important practical value to be able to map out these
link regimes precisely and understand which techniques are required in
each regime.
These link regimes can serve as a basis for performing link adaptation of the composite scheme, depending on the available channel
state information. Specifically, the transmitters only need to know certain relative order of the link strengths in order to
determine what combination of techniques is optimal.

Mapping out these link regimes analytically is challenging. It is mainly because the associated
rate region optimization problem contains multiple variables and,
in general, is non-convex. In order to simplify the analysis, in this paper, we focus only on a subset of the considered techniques including
the partial DF relaying and independent coding. Analyzing block Markov coding with
independent coding is done separately in \cite{lisa2014, lisa2015}, and the synthesis of all three techniques is left as a future work.
%
\subsection{Link regimes of partial decode-forward with independent relay coding}
For this analysis, we consider only  independent coding and
partial relaying. Each user splits its message into two parts, the
relay decodes only one part -- the common message -- from each user,
then puts the decoded message pair into a corresponding bin and
forwards the bin index in the next block. The relay effectively performs
network coding where the signal from the relay contains the common messages of both users but is independent of signals
from both users. While there is no coherent beamforming gain between the
relay and each user, the total relay transmit power is {\it effectively}
allocated to carry the common message of only one user, since the
other common message is already known at the decoding user.

This independent pDF scheme can be viewed as a combination of DF
relaying (of the common message parts) and direct transmission (of the
private message parts). For the standard one-way Gaussian relay channel with single antenna at each node,
DF relaying is known to be  more useful than direct transmission only if the source-relay
link is stronger than the source-destination link. Furthermore, pDF
performs no better than DF relaying \cite{gmze}. The main reason is the
power split required to transmit two different information parts at
the source. If the source-relay link is stronger than the
direct link, then it is more beneficial to
put all power into the signal part to be decoded via the source-relay
link, practically leading to no message splitting nor partial DF.

For the TWRC, is
partial DF any better than full DF, and is it more than just a mere
linear (time-shared) combination of full DF and direct transmission?
We find that the answers to these questions are both positive. Indeed
independent pDF can outperform full DF and achieve rate points
strictly outside the DF-DT time-shared region under certain link conditions.

To obtain this result,
we analyze the rate
region in \eqref{rates-constr} directly, after setting
$\alpha_1=\alpha_2=0$ (no block Markov coding). The scheme reduces to full
DF when the optimal solution is $\gamma_1^\star = \gamma_2^\star = 0$, and
to direct transmission (DT) when $\beta_1^\star = \beta_2^\star=0$. When neither case holds, partial DF
relaying is used. The results are summarized next.
\begin{thm}
\label{thm:4_cases}
The independent partial DF (pDF) scheme can be optimally divided into
several operating regimes depending on the link states as follows.
\begin{enumerate}
\item[A)] One user performs pDF and the other user
  switches between  DT and full DF when
  \vspace{-2mm}
  \begin{subequations}\label{case1}
  \begin{align}
    g_{21}^2&<g_{r1}^2< \left(g_{21}^2+g_{2r}^2\frac{P_r}{P_1}\right)(1+g_{r2}^2P_2), \nonumber\\
    & \text{and} \quad g_{r2}^2 < g_{12}^2,
    \label{case1:a} \\
    \text{or} \quad g_{12}^2&< g_{r2}^2<\left(g_{12}^2+g_{1r}^2\frac{P_r}{P_2}\right)(1+g_{r1}^2P_1), \nonumber\\
    & \text{and} \quad g_{r1}^2 < g_{21}^2.
    \label{case1:b}
  \end{align}
  \end{subequations}
  This regime further contains two sub-regimes. The sub-regime in which the pDF scheme achieves the time-shared region between DF and DT
  is specified by
  %
  \begin{subequations}\label{case2}
      \begin{align}
        \!\!\!\!\!\!\!\!\!\!\!\!g_{21}^2&<g_{r1}^2<\min\left\{\left(g_{21}^2+g_{2r}^2\frac{P_r}{P_1}\right),\;g_{21}^2(1+g_{r2}^2 P_2)\right\},\;
        \nonumber\\
        &g_{r2}^2 < g_{12}^2,\;
        \text{and}\;R_{S,1}\leq R_{TS,1}\label{case2:b}\\
        \!\!\!\!\!\!\!\!\!\!\!\!\text{or}\; g_{12}^2&<g_{r2}^2<\min\left\{\left(g_{12}^2+g_{1r}^2\frac{P_r}{P_2}\right),\;g_{12}^2(1+g_{r1}^2 P_1)\right\},\nonumber\\
        & g_{r1}^2 < g_{21}^2,\;
        \text{and}\;R_{S,2}\leq R_{TS,2}\label{case2:e}
      \end{align}
    \end{subequations}
    where $R_{S,i}$ and $R_{TS,i}$ (for $i\in{1,2}$) are the maximum sum rates
  of the pDF scheme and DF-DT time-shared scheme, respectively, as given in
  (\ref{ffft}). Outside this sub-regime, the pDF scheme outperforms both independent DF and DT by
achieving rate pairs strictly outside the time-shared region of these two schemes.
\item[B)] One user performs full DF and the other user
  performs DT when 
    \begin{subequations}\label{case12}
    \begin{align}
      g_{r1}^2&> \left(g_{21}^2+g_{2r}^2\frac{P_r}{P_1}\right)(1+g_{r2}^2P_2), \nonumber\\
        & \text{and} \quad g_{r2}^2 < g_{12}^2, \label{case12:c}\\
      \text{or}\quad g_{r2}^2&>\left(g_{12}^2+g_{1r}^2\frac{P_r}{P_2}\right)(1+g_{r1}^2P_1),\nonumber\\
        & \text{and} \quad g_{r1}^2 < g_{21}^2.\label{case12:d}
    \end{align}
  \end{subequations}
  Here pDF always achieves rate pairs strictly outside the DF-DT time-shared region.
\item[C)] Both users perform partial DF relaying when
  \begin{align}
    g_{r1}^2 &>  g_{21}^2, \quad
    g_{r2}^2 >  g_{12}^2 \nonumber\\
    \text{and}& \quad
    C(g_{r1}^2P_1+g_{r2}^2P_2) < C(g_{21}^2P_1)+C(g_{12}^2P_2)
    \label{case5}
  \end{align}
In this regime, the
  pDF scheme achieves exactly the DF-DT time-shared region.
\item[D)] Both users perform full DF relaying when  
  \begin{align}
    g_{r1}^2 &>  g_{21}^2, \quad
    g_{r2}^2 >  g_{12}^2 \nonumber\\
     \text{and} &\quad
    C(g_{r1}^2P_1+g_{r2}^2P_2) \geq C(g_{21}^2P_1)+C(g_{12}^2P_2)
    \label{case3}
  \end{align}
Here pDF
  reduces to full DF relaying, which is strictly better than DT.
\item[E)] Both users perform DT without using the
  relay when  
  \begin{align}
    g_{r1}^2 \leq g_{21}^2, \quad  \text{and} \quad
    g_{r2}^2 \leq g_{12}^2.
    \label{case4}
  \end{align}
 Here pDF reduces to DT which
  is strictly better than full DF. Neither user uses the relay.
\end{enumerate}
\end{thm}
\begin{proof}
  See Appendix B for details.  The next section summarizes the approach. 
\end{proof}
\begin{figure*}
\begin{align}\label{ffft}
R_{S,1}&=C\left(g_{12}^2\gamma_2^{\ast}\right)
+C\left(\frac{g_{r2}^2(P_2-\gamma_2^{\ast})}{1+g_{r1}^2P_1+g_{r2}^2\gamma_2^{\ast}}\right)+
\mu_1C\left(\frac{g_{r1}^2P_1}{1+g_{r2}^2\gamma_2^{\ast}}\right) \\
\text{where} \quad \mu_1&=\frac{C(g_{12}^2P_2)-
C\left(\frac{g_{r2}^2P_2}{1+g_{r1}^2P_1}\right)}{C(g_{r1}^2P_1)-C(g_{21}^2P_1)},
\quad
\gamma_2^{\ast}=\left(\frac{g_{r2}^2g_{r1}^2\mu_1 P_1-(g_{12}^2-g_{r2}^2+g_{12}^2g_{r1}^2P_1)}
{g_{r2}^2g_{12}^2-g_{r2}^4+g_{r2}^2g_{12}^2g_{r1}^2P_1(1-\mu_1)}\right)^+, \nonumber \\
R_{TS,1}&=\frac{C(g_{r1}^2P_1)C(g_{12}^2P_2)
-C\left(\frac{g_{r2}^2P_2}{1+g_{r1}^2P_1}\right)C(g_{21}^2P_1)}{C(g_{r1}^2P_1)-C(g_{21}^2P_1)},\nonumber
\end{align}
\caption{Values for regime A (conditions (\ref{case2:b})) in Theorem \ref{thm:4_cases}. Here $\mu_1$ is
  the slope of the time-shared line between DF and DT rate
  regions, and $\gamma_2^{\ast}$ is the optimal $\gamma_2$ that maximizes
  $R_{S,1}$. $(x)^+$ denotes the positive part of $x$ where $(x)^+=x$ if $x\geq 0$ while $(x)^+=0$ if $x< 0$. Other values $R_{S,2}$, $R_{TS,2},$ $\mu_2$ and $\gamma_1^{\ast}$ are
  defined similarly by switching all subscripts from $1\rightarrow 2$ and
  $2\rightarrow 1$.}
  \vspace*{-6mm}
\end{figure*}
\vspace{-6mm}
\subsection{Analysis approaches}
  The conditions and transmission schemes in Theorem \ref{thm:4_cases} are obtained by comparing the rate regions of independent partial DF in (\ref{rates-fd})  with those of full DF and direct transmission.
   This comparison is performed by first examining the corner points of respective rate regions. It is quite straightforward to show that
   the partial DF scheme achieves a corner point outside the DF-DT time-shared region in several sub-cases of A and B.
\begin{figure}[t]
    \begin{center}
    \includegraphics[width=0.45\textwidth]{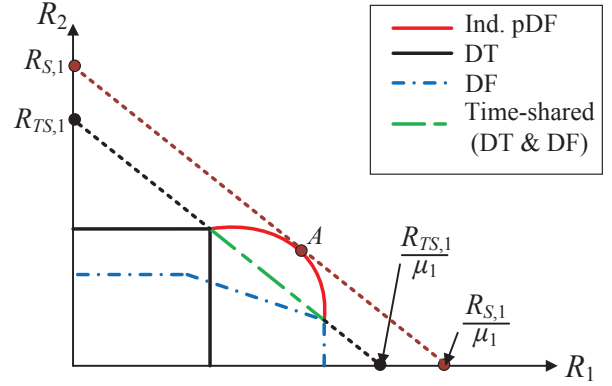}
    \caption{Illustration for comparison between independent pDF relaying and DF-DT time sharing for Regime A in Thm.  \ref{thm:4_cases}.}
    \label{figAAP}
    \end{center}
    \vspace{-6mm}
\end{figure}
\begin{figure}[t]
    \begin{center}
    \includegraphics[width=0.48\textwidth, height=58mm]{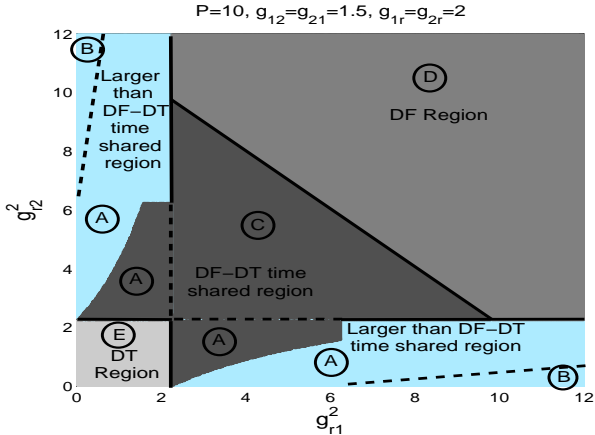}
    \caption{Link-state  regions for the independent pDF
      scheme  as specified in Theorem \ref{thm:4_cases}.}
    \label{figoprg}
    \end{center}
    \vspace{-6mm}
\end{figure}

    For the cases where the corner points of the pDF region coincides with the corner points of the DF-DT time-shared region,
    we examine the maximum weighted sum rate curve of the pDF scheme to see if there is a point outside the DF-DT time-shared region
    as illustrated in Figure \ref{figAAP}. Specifically, we consider an optimization problem of the weighted sum-rate point (point A), where the weights
    are the same as the coefficients of the slope of the outermost time-sharing line between DF and DT (the dashed green line). The slope coefficient $\mu_1$ is obtained based on the easily computed  corner points of the DF and DT schemes  as shown in (\ref{ffft}). We then maximize the weighted sum
    rate $(R_2+\mu_1 R_1)$ subject to  rate constraints of the pDF scheme. If the maximum weighted sum
    rate point is outside the DF-DT time-sharing line $(R_{S,1}>R_{TS,1})$, then the pDF scheme strictly outperforms time sharing between DF and DT.

    The optimization of the weighted sum rate also reveals the structure of the optimal input signals of the two users. For case A under conditions (\ref{case1:a}), this optimization for every given power allocation of user 2's private message results in the optimal power allocation of user 1's private message as either $0$ or $P$. Hence, the optimal signal structure for one user is pDF while for the other user is either full DF or DT. The opposite structure holds for conditions (\ref{case1:b}).

    For  case C, we show that full DF scheme maximizes the individual rates of pDF scheme while DT maximizes the sum rate.
    The rest of the region is obtained by time sharing of DF-DT schemes. For cases D and E, it is easy to show that the optimal power allocations
    for both private messages are either $0$ or $P$ which reduces the scheme to full DF (case D) or DT (case E).
\vspace{-4mm}
\subsection{Discussion of link regimes}
The link regimes of Theorem \ref{thm:4_cases} are graphically illustrated in Figure \ref{figoprg}.
Denote SNR$_{ij}=g_{ij}P_j$ as the received SNR at use $i$ for signals from node $j$, ignoring any interference or signals from other nodes. The link-state regimes in Theorem 2 can then be expressed in terms of these received SNRs. First, consider
regime A under condition (\ref{case1:a}) (condition (\ref{case1:b}) is similar with the user indices swapped).
This condition corresponds to user 1's relay link slightly stronger than its direct link, SNR$_{r1}>$SNR$_{21}$.
 However, when considering the interference at the relay from user $2$, the received SNR at the relay is less than the received SNR at user $2$ from both user $1$ and the relay, i.e., $\frac{\text{SNR}_{r1}}{1+\text{SNR}_{r2}}<$ $\text{SNR}_{21}+\text{SNR}_{2r}$. Then, the optimal transmission scheme is for user $2$ to perform pDF relaying while user $1$ switches between direct transmission and full DF relaying. The switching transmission from user $1$ depends on the power allocated to the private part of user $2$ as it determines the amount of interference at the relay and hence whether direct transmission or full DF relaying is preferred at user $1$. Performing pDF relaying for user
2 provides two benefits. It allows user 2 the option of achieving  its
highest individual rate obtained in this case from direct transmission
which is a special case of the pDF relaying. But it also allows user $2$ the option of reducing  the
interference to user 1's information at the relay since it allows the
relay to decode part of user 2's information, which effectively removes this
portion of user 2's signal from the interference, in order to decode
user 1's information.

In regime B (conditions (\ref{case12:c})), user 1's relay link is
 significantly stronger than its direct link such that even when considering the interference at the relay from user $2$, the received SNR at the relay from user $1$ is greater than the received SNR at user $2$ from both user $1$ and the relay, i.e., $\frac{\text{SNR}_{r1}}{1+\text{SNR}_{r2}}>$ $\text{SNR}_{21}+\text{SNR}_{2r}$. Therefore, the hybrid scheme is optimal in this case where user $1$ always performs full DF relaying while user $2$ performs DT.  This hybrid scheme gives user 2 its maximum rate in this region while presents no harm to user 1's rate even under the maximum interference from user $2$ at the relay because of the strong link from user $1$ to relay.

 Regime C indicates that when the received SNR at the relay from each user is greater than the received SNR at the other user while the total sum rate at the relay from both users is less than the sum of both DT rates, then the DF-DT time-shared scheme can achieve the same rate region as the pDF scheme. However, in regime D, when the total sum rate at the relay from both users is greater than the sum of both DT rates, then full DF relaying becomes the optimal scheme since both user-to-relay links are strong enough to support full DF relaying. Regime E indicates that when the received SNR at the relay from each user is less than the
corresponding direct SNR, the relay should not be used at all,
similar to the one-way relay channel.
%

Independent pDF relaying is rate-optimal in Regimes
A, B and C. Under conditions \eqref{case2} and \eqref{case5}, pDF achieves the same rate
region as DF-DT time-sharing; that is, it effectively implements time-sharing in these cases.
Under conditions \eqref{case1} and \eqref{case12}, however, pDF strictly outperforms DF-DT time-sharing.
\section{A half-duplex decode-forward scheme}\label{sec: HDS}
In this section, we propose a DF scheme that utilizes coherent
relaying, independent relaying and partial relaying for the
half-duplex mode. Different from the full-duplex scheme, in
the half-duplex mode, each node can only either transmit or receive at
each time in the same frequency band. We consider a time-division
implementation for the half-duplex mode and propose a $6$-phase scheme as depicted in Figure \ref{fig:6_phases}(b).
\vspace{-4mm}
\subsection{Transmission scheme}
\label{tx-scheme}
There are several important differences between the half-duplex
and full-duplex schemes. First, because of half-duplex transmission,
block Markov superposition coding changes to phase superposition
coding. In  each block, each user transmits only new information and
does not repeat old information of the previous block. Each user
divides its message into two parts: a private part and a common part,
and  superposes the private message on the common one. But not all
messages are transmitted in each phase. During a specific phase,
if only the relay is receiving, then each user transmits only its
common message; otherwise, each user transmits both common and
private messages. The
superposition coding between phases creates a
coherent beamforming in certain phases (when a user and the relay
both transmit) in a way similar to the full-duplex beamforming created by block Markovity.
Second, in each phase, the relay performs only independent coding or
coherent coding, but not both. When forwarding both common messages,
the relay uses independent coding to send the bin index of this
message pair. When forwarding only one common message, the relay uses
coherent coding.
Third, decoding at each node is performed simultaneuosly across several
phases. The relay decodes the common messages at the end of phase $3$,
by simultaneously using signals received in all first three
phases. Each user decodes the other user's messages at the end of
phase $6$ using all its received signals in its receiving
phases. Since user decoding occurs at the end of the same block that
information is transmitted, the decoding delay is shorter than in the
full-duplex scheme by one block.

Next we describe the signal construction in each phase
and decoding rule at each node. 

\subsubsection{Transmit signal construction}
Since block index is unnecessary in the half-duplex scheme, we will
drop it from the notation. Again let the message of user $i$ be $m_i$ for $i\in\{1,2\}$. Each user splits its message as
$m_{i} = (m_{i0},m_{ii})$, where $m_{i0}$ is the common message and $m_{ii}$ is the private
message. The transmit signals during each phase are as follows:
\begin{align}
  \text{Phase 1:} \quad
  X_{11}&=\sqrt{\alpha_{11}}U_1^{(1)}(m_{10})+\sqrt{\beta_{11}}V_1^{(1)}(m_{11})
  \nonumber\\
  \text{Phase 2:} \quad
  X_{22}&=\sqrt{\alpha_{22}}U_2^{(2)}(m_{20})+\sqrt{\beta_{22}}V_2^{(2)}(m_{22})
  \nonumber\\
  \text{Phase 3:} \quad X_{13}&=\sqrt{\alpha_{13}}U_1^{(3)}(m_{10}),\quad
  X_{23}=\sqrt{\alpha_{23}}U_2^{(3)}(m_{20})\nonumber\\
  \text{Phase 4:}  \quad
  X_{14}&=\sqrt{\alpha_{14}}U_1^{(4)}(m_{10})+\sqrt{\beta_{14}}V_1^{(4)}(m_{11}),\nonumber\\
  X_{r4}&=\sqrt{\gamma_{34}}U_1^{(4)}(m_{10})\nonumber\\
  \text{Phase 5:} \quad
  X_{25}&=\sqrt{\alpha_{25}}U_2^{(5)}(m_{20})+\sqrt{\beta_{25}}V_2^{(5)}(m_{22}),\nonumber\\
  X_{r5}&=\sqrt{\gamma_{35}}U_2^{(5)}(m_{20})\nonumber\\
  \text{Phase 6:} \quad X_{r6}&= \sqrt{{\gamma}_{36}}W(m_{10},m_{20}),
  \label{hd-signaling}
\end{align}
where $U_i^{(k)}$ and $V_i^{(k)}$ ($i=1,2, k =1\ldots 6$) respectively
represent the signal for the common message and the private message
of user $i$ in phase $k$, and $W$ represents the signal for the bin
index of the common message pair decoded at the relay. All signals
$U_i^{(k)},V_i^{(k)},W$ are independent ${\cal CN}(0,1)$. Note that the signals encode the same
information but in different phases are independent of each other (for
example, $U_i^{(k)}$ all encodes $m_{i0}$ but are independent of each
other  for different values of $k$); however, the signals encoding the same information in the same
phase are fully correlated (for example, in phase 4, user 1 and the
relay both send the same signal $U_1^{(4)}$ that encodes $m_{10}$,
just with different power allocation). Since there is no risk of
ambiguity, we will drop phase superscripts in the subsequent
discussion.

The power allocation factors $\alpha_*,\beta_*,\gamma_*$ satisfy the following power constraints:
\begin{align}\label{half-power-const}
  \tau_1(\alpha_{11}+\beta_{11}) + \tau_3 \alpha_{13} +
  \tau_4(\alpha_{14}+\beta_{14})& \leq P_1\nonumber\\
  \tau_2(\alpha_{22}+\beta_{22}) + \tau_3 \alpha_{23} +
  \tau_5(\alpha_{25}+\beta_{25})& \leq P_2,\nonumber\\
  \tau_4 \gamma_{34} + \tau_5 \gamma_{35} + \tau_6{\gamma}_{36} & \leq P_r.
\end{align}
Note that the above constraints include power control with regards to
the phase durations $\tau_{*}$ to ensure that each block has a
constant average power.
\subsubsection{Decoding rules}
Again we use ML decoding at each node. The important feature here is
that each node uses the received signals in several
phases simultaneously. Specifically, the relay is receiving in
phases 1, 2 and 3; thus it decodes at the end of phase 3 by using the
received signals in all first three phases. Similarly, user 1 decodes
at the end of phase 6 by using its received signals in phases $(2, 5,
6)$; as so does user 2 using received signals in phases $(1, 4, 6)$.

At the relay, the received signals in the first three phases are
\begin{align}
  Y_{r1}&=h_{r1}\sqrt{\alpha_{11}}U_1+h_{r1}\sqrt{\beta_{11}}V_1+Z_{r1},
  \nonumber\\
  Y_{r2}&=h_{r2}\sqrt{\alpha_{22}}U_2+h_{r2}\sqrt{\beta_{22}}V_2+Z_{r2}
  \nonumber\\
  Y_{r3}&=h_{r1}\sqrt{\alpha_{13}}U_1+h_{r2}\sqrt{\alpha_{23}}U_2+Z_{r3}.
\end{align}
The relay is interested in decoding $(U_1,U_2)$ to obtain messages
$(m_{10},m_{20})$ and treats $(V_1,V_2)$ as noise in this process. The
optimal relay decoding rule can be written as follows:
\begin{align}
  \left( m_{10}, m_{20} \right) = \arg \max
  P(U_{1}|Y_{r1}) \cdot P(U_{2}|Y_{r2}) \cdot P(U_{1},U_{2}|Y_{r3}).
  \label{relay-dec-hd}
\end{align}
Note here the decoding probability is joint (as a product) across
three phases.

Decoding at the users is similar. Since operations at two users are
alike, we will only describe the decoding rule for user 2. User 2 is
in receiving mode in phases ($1, 4, 6$) with received signals:
\begin{align}
   Y_{21}&=h_{21}\sqrt{\alpha_{11}}U_1+
   h_{21}\sqrt{\beta_{11}}V_1+Z_{21}, \nonumber\\
   Y_{24}&=h_{21}\sqrt{\alpha_{14}}U_1 + h_{21}\sqrt{\beta_{14}}V_1 +
   h_{2r}\sqrt{\gamma_{34}}U_1+Z_{24}. \nonumber\\
   Y_{26}&=h_{2r}\sqrt{{\gamma}_{36}}W+Z_{26},
\end{align}
User 2 jointly decodes both messages $(m_{10},m_{11})$ at the end of
phase $6$ by the following rule:
\begin{align}
  \left( m_{10}, m_{11} \right) = \arg \max &
  P(U_{1},V_1|Y_{21}) \cdot P(U_{1},V_1|Y_{24}) \nonumber\\
  &\;\cdot P(W|Y_{26},m_{20}).
  \label{user-dec-hd}
\end{align}
Again the decoding probability is a product across three phases, and
in the last phase, the decoding rule utilizes side information of
$m_{20}$ available at user 2.
\vspace{-4mm}
\subsection{Achievable rate region}
With the above transmit signals and decoding rules, we have the
following result.
\begin{thm}
\label{thm:gaussian_half_duplex}
The following rate region is achievable for the half-duplex Gaussian
two-way relay channel with the above  6-phase partial decode-forward transmission
scheme:
\begin{align} \label{hd-rate-region}
  R_1 & \leq \min \{J_1+J_4,J_5 \},\nonumber\\
  R_2 & \leq \min \{J_2+J_6,J_7 \},\nonumber\\
  R_1 + R_2 & \leq J_3 +J_4 + J_6,
\end{align}
where
\begin{align*}
  J_1 &
  =\tau_1 C\left(\frac{g_{r1}^2\alpha_{11}}{g_{r1}^2\beta_{11}+1}\right)
  +\tau_3 C(g_{r1}^2\alpha_{13}), \nonumber \\
  J_2 &= \tau_2
  C\left(\frac{g_{r2}^2\alpha_{22}}{g_{r2}^2\beta_{22}+1}\right)+\tau_3
  C(g_{r2}^2\alpha_{23}) \nonumber \\
  J_3 &=\tau_1
  C\left(\frac{g_{r1}^2\alpha_{11}}{g_{r1}^2\beta_{11}+1}\right)+\tau_2
  C\left(\frac{g_{r2}^2\alpha_{22}}{g_{r2}^2\beta_{22}+1}\right) \nonumber\\
  &\;+ \tau_3 C(g_{r1}^2\alpha_{13}+g_{r2}^2\alpha_{23})\nonumber\\
  J_4 &=\tau_1C(g_{21}^2\beta_{11}) + \tau_4C(g_{21}^2\beta_{14}),\nonumber\\
  J_5 &= \tau_1C\left(g_{21}^2(\alpha_{11}+\beta_{11})\right)+ \tau_6
  C(g_{2r}^2\gamma_{36})  \nonumber \\
  &\;+\tau_4C\left((g_{21}\sqrt{\alpha_{14}}+g_{2r}\sqrt{\gamma_{34}})^2
    +g_{21}^2\beta_{14}\right)
  \nonumber \\
  J_6 &=\tau_2C(g_{12}^2\beta_{22}) + \tau_5C(g_{12}^2\beta_{25})  \nonumber \\
  J_7 & = \tau_2C\left(g_{12}^2(\alpha_{22}+\beta_{22})\right)+\tau_6
  C(g_{1r}^2\gamma_{36}) \nonumber \\
  &\;+ \tau_5C\left((g_{12}\sqrt{\alpha_{25}}+g_{1r}\sqrt{\gamma_{35}})^2
    +g_{12}^2\beta_{25}\right) \nonumber
  \end{align*}
and phase duration $\sum_i \tau_i = 1$, power allocation parameters satisfy the
constraints in \eqref{half-power-const}.
\end{thm}
\begin{IEEEproof}[Sketch of the proof]
  The rate region follows directly from the signal construction and
  decoding rules in Section \ref{tx-scheme}. In particular, decoding
  of common messages at the relay produces rate terms $J_1,J_2,J_3$ as
  in a multiple access channel. Decoding of the common and private
  messages of user 1 at user 2 produces rate terms $J_4$ and $J_5$ as
  a direct result of the superposition coding structure. Similarly,
  $J_6$ and $J_7$ result from the decoding at user 1. Note that all
  these rate expressions contain multiple terms that account for the
  fact that a message is sent, received and decoded simultaneously over
  multiple phases.
\end{IEEEproof}
\begin{rem}\label{mh1}
Employing sequential decoding for the common and then private message parts instead of joint decoding of both message parts  can lead to the same achievable rate region. However, simultaneous  decoding among the phases must also be employed. Specifically, we can achieve the same rate region  as in (\ref{hd-rate-region}) if user $2$ sequentially decodes the common message part $(m_{10})$ and then the private message part $(m_{11})$  instead of jointly decoding them, provided that user $2$ decodes $m_{10}$ simultaneously  from all received signals in phases $1,$ $4$ and $6$ and then decodes $m_{11}$  simultaneously from all received signals in phases $1$ and $4$.
\end{rem}
\vspace{-4mm}
\subsection{Comparison with existing schemes}\label{hdincls}
In half-duplex rate region \eqref{hd-rate-region}, by setting the
appropriate phase durations to $0$, we can recover existing results.
By setting $\tau_4=\tau_5=0$, the proposed $6$-phase scheme
   reduces to the 4-phase scheme in our previous work \cite{zhong2012partial}, which
   encompasses  the 2, 3 and 4-phase DF schemes in \cite{kim2008per}.

The proposed scheme is simpler than the 6-phase scheme in \cite{gerdes} as each user splits its message into $2$ parts only instead of $6$ parts as in \cite{gerdes}. In practice, fewer split parts lead to fewer parameters to be optimized for performance. Although sequential decoding used in \cite{gerdes} is simpler than joint decoding, it leads to a smaller achievable rate region. We also show in Remark \ref{mh1} that sequential decoding can be used to first decode the common and then the private message part provided that the decoding is performed simultaneously  among the phases. The rate region of our scheme in  Theorem \ref{thm:gaussian_half_duplex} strictly encompasses the rate region of the scheme in \cite{gerdes} by setting $\beta_{11}=\beta_{22}=0$, as also shown in numerical examples in Figure \ref{fig:HDFIGTW}.

Our scheme is also  simpler in terms of the message splitting than the 6-phase scheme in \cite{ASHARA, nw6p}, in which  each user splits its message into $3$ parts. However,   there is no superposition coding or coherent transmission in this scheme which simplifies its signaling.
Specifically, in phases $4$ and $5$, each user and the relay transmit independent signals non-coherently to the other user.
The rate region of the scheme in \cite{ASHARA, nw6p} can be obtained from Theorem \ref{thm:gaussian_half_duplex}
as a special case by setting $\beta_{11}=\beta_{22}=\alpha_{14}=\alpha_{25}=0$, and  is smaller than our achievable rate region.
\section{Numerical results}\label{sec: NR}
In this section, we provide numerical examples for the rate region
achieved by the proposed full and half-duplex schemes and compare it
with existing schemes. These regions are obtained by taking the convex closure of all rate regions obtained with all possible resource allocations,
 including powers and phases \cite{gamal2010lecture}. We also provide numerical illustration for the
optimal link regimes as shown in Theorem \ref{thm:4_cases}. In all figures, we set $P_1=P_2=P_r=P$.

Figures \ref{fig:case1} compares
between the achievable rate regions in the
full-duplex mode for proposed independent and coherent pDF scheme in Theorem \ref{thm:complete_pDF_gaussian}, existing schemes and the cut-set outer bound  under channel conditions specified in the figure. Results show that under such channel
conditions, the proposed scheme outperforms all existing schemes and
their time-sharing.
For the independent pDF scheme, these results are in agreement with Theorem \ref{thm:4_cases}. The
channel configurations in Figure \ref{fig:case1} corresponds to Regime
$B$ with condition (\ref{case12:c}), under which the composite scheme outperforms time-sharing of existing schemes.

Similar results are obtained for the half-duplex schemes in Figure
\ref{fig:HDFIGTW} with link qualities shown on top of  the figure. The asymmetric links between the two users can appear in practice in FDD systems.
The proposed pDF scheme includes as special cases the $4$-phase
schemes in \cite{kim2008per}\cite{zhong2012partial}. Note that our previous $4$-phase scheme in \cite{zhong2012partial} achieves a rate region very close to that of the $6$-phase composite scheme for the chosen channel configuration. The composite $6$-phase scheme also  outperforms the existing $6$-phase schemes \cite{gerdes, ASHARA, nw6p}.
The performance improvement in our scheme comes from
the simultaneous  decoding among all phases instead of separate decoding at
each phase as in \cite{gerdes} or only among a few phases as in
\cite{ASHARA, nw6p}. Furthermore, the schemes in \cite{ASHARA, nw6p}
do not have coherent relaying.
%
\begin{figure}[t]
    \begin{center}
    \includegraphics[width=0.45\textwidth, height=55mm]{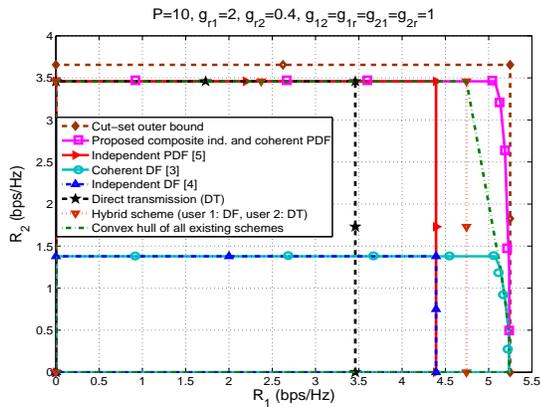}
    \caption{Full-duplex composite pDF scheme achieves rates outside the time-shared region of
      existing schemes.}
    \label{fig:case1}
    \end{center}
    \vspace{-6mm}
\end{figure}
\begin{figure}[!t]
    \begin{center}
    \includegraphics[width=0.45\textwidth, height=55mm]{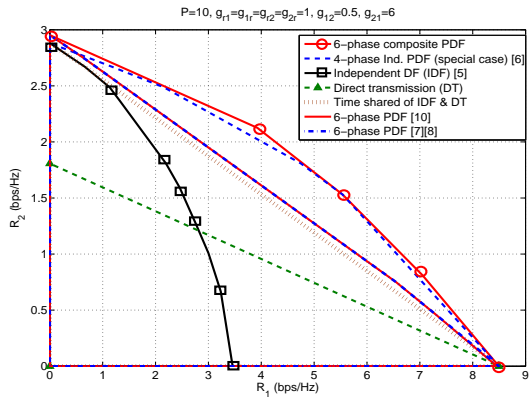}
    \caption{Rate region comparison between the proposed $6$-phase composite and existing
      schemes in the half-duplex mode.}
    \label{fig:HDFIGTW}
    \end{center}
    \vspace{-6mm}
\end{figure}

Figure \ref{fig:opregions} illustrates the different geometric
link regimes on a 2$D$ plane for the independent pDF scheme as
described in Theorem \ref{thm:4_cases}. The locations of user 1 and
user $2$ are fixed at $(-1,0)$ and $(1,0)$, respectively, while the
relay can be anywhere on the plane. The path-loss model determines the
channel gain between any two nodes based on their relative distance
as $g_{ij}=d_{ij}^{-\gamma/2}$ with $\gamma=2.4,$ which is an example  for indoor propagation at $900$MHz \cite{raport}.

The results of Figure \ref{fig:opregions} match the results and implications of Theorem \ref{thm:4_cases}. Figure \ref{fig:opregions} shows that full DF relaying is optimal when the relay is
close to both users and has strong links to them such that it can
fully decode their information (case $D$ in Theorem \ref{thm:4_cases}). The hybrid scheme with full DF relaying from one user and
DT from the other user is optimal when the relay
is very close to one user and far from the other (case $B$ in Theorem \ref{thm:4_cases}). Switching between full DF relaying and DT from one user and performing
pDF from the other user is optimal when the relay
is slightly close to one user and far from the other (case $A$ in Theorem \ref{thm:4_cases}). Partial DF relaying from both users is optimal when the relay
is slightly far from both users, where   pDF relaying obtains the same
performance as DF-DT time sharing (case $C$ in Theorem \ref{thm:4_cases}).
DT from both users is optimal when the relay is too far from both users such that
decode-and-forward degrades performance (case $E$ in Theorem \ref{thm:4_cases}). As the relay moves from afar towards a user, the transmission strategy for that user changes from DT to switching between DF and DT, then to full DF when the relay is very close to that user, and to pDF when the relay is moving further away again. Figure \ref{fig:opregions} shows that in the regions
where the use of relay is viable, independent pDF relaying can
improve performance upon existing schemes for a large portion of these
regions.
\begin{figure}[t]
    \begin{center}
    \includegraphics[width=0.45\textwidth, height=55mm]{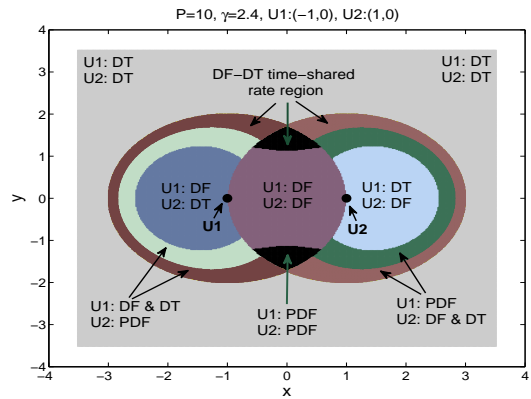}
    \caption{ Geometric regions of relay location for the optimal use
      of different sub-schemes. Some regions overlap when the rate performance is the same.}
    \label{fig:opregions}
    \end{center}
    \vspace{-6mm}
\end{figure}
\begin{figure}[!t]
    \begin{center}
    \includegraphics[width=0.45\textwidth, height=55mm]{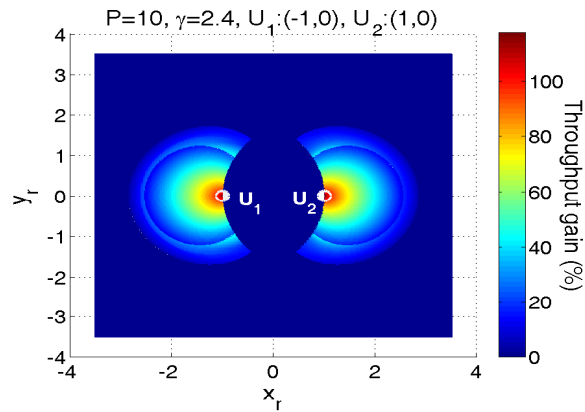}
    \caption{Throughput gain obtained by independent pDF compared with DF-DT time sharing for the same spatial regions in Figure \ref{fig:opregions}.}
    \label{fig:ThroG}
    \end{center}
    \vspace{-6mm}
\end{figure}

Figure \ref{fig:ThroG} shows the throughput gain obtained by using the independent pDF scheme compared with the DF-DT time-shared
scheme for the same spatial regions as in Figure \ref{fig:opregions}. Specifically, 
the throughput gain $(\Gamma)$ is obtained as $\Gamma=\frac{R_{S,1}-R_{TS,1}}{R_{TS,1}}100\%$. Results show that in the pDF region in Figure \ref{fig:opregions} (the black and green regions), independent
pDF relaying can improve the throughput by up to $40\%$. Moreover, in the hybrid DF-DT region in Figure \ref{fig:opregions} (the two blue regions), the hybrid scheme can
 improve the throughput by more than $100\%$.
 These gains are significant, especially considering that in the hybrid regions, the optimal strategy is quite simple as having only one user use the relay in the full DF mode while the other performs DT. Even though there is no pDF relaying in the hybrid regions, the knowledge that this hybrid scheme is optimal comes as a direct result of link-regime analysis of pDF relaying. These results emphasize the importance of adapting the transmission schemes according to the link regimes.
\section{Conclusion}\label{sec: CN}
We have proposed DF-based schemes for the full and half-duplex
TWRC that utilize coherent relaying, independent relaying, and
partial relaying techniques. The proposed schemes outperforms all
existing schemes by achieving a larger rate region. We further analyze
independent and partial relaying in detail and analytically
derive the link regimes for the optimal use of these techniques. These
regimes show when it is necessary to use independent partial DF
relaying, and when it is sufficient to use simpler schemes such as
full DF relaying or just direct transmission, in order to get the
largest possible rate region. Contrast to established results for the
one-way relay channel, we show that partial DF relaying can outperform
full DF relaying in the TWRC. Independent and partial
DF relaying achieves a strictly larger rate region than the DF-DT
time-shared region when the relay
has strong link with one user and weak link with the other user, or
when it has links to both users that are just marginally stronger than
their direct links. These link regimes are useful for adapting
 the proposed schemes to practical  link configurations.
\appendices
\section*{Appendix A: Proof of Theorem \ref{thm:complete_pDF_gaussian}}
\label{app_partial_alldf}
\setcounter{subsubsection}{0}
In this appendix, we provide the proof of Theorem
\ref{thm:complete_pDF_gaussian} through information theoretic analysis
of a discrete memoryless channel specified by a collection of pmf
$p(y_1,y_2,y_r|x_1,x_2,x_r)$, where $x_k, y_k$ ($k=1,2,r$, $r$ stands
for the relay) are the input and output signals of node $k$.
We define a $(2^{nR_1},2^{nR_2},n,P_e)$ code based on standard definitions as in \cite{gamal2010lecture}.
%
%
\subsubsection{Codebook generation}
The codebook is generated according to the joint distribution
\begin{align}
  P&=p(w_1)p(u_1|w_1)p(x_1|u_1,w_1)p(w_2)p(u_2|w_2)\nonumber\\
  &\;\;\times p(x_2|u_2,w_2)p(x_r|w_1,w_2)
  \label{code-distr}
\end{align}
The messages of each user are encoded
independently by superposition coding. For each block $j\in [1:b]$,  according to $p(w_1)p(u_1|w_1)p(x_1|u_1,w_1)$, independently
generate $2^{nR_{10}}$ codewords $w_1^n(m_{10,j-1})$, $2^{nR_{10}}$ codewords $(u_1^n(m_{10,j}|m_{10,j-1}))$ and $2^{nR_{11}}$
codewords $(x_1^n(m_{11,j}|m_{10,j},m_{10,j-1}))$ that encode $m_{10,j-1},$ $m_{10,j}$ and $m_{11,j}$, respectively.
Similarly for user 2's  messages.

The relay only encodes the
common messages of both users in the previous block by using both
superposition coding and random binning. First, equally and uniformly
divide common message pairs $(m_{10}, m_{20})$ into $2^{nR_{r}}$
bins and let $K(m_{10}, m_{20})$ denote the bin index. Then superpose
codewords for bin indices on the codeword generated at each user for
each common message. Specifically, for each pair of previous common
messages $(m_{10,j-1},m_{20,j-1})$, generate $2^{nR_{r}}$ sequences
$x_r^n(K(m_{10,j-1},m_{20,j-1})|m_{10,j-1},m_{20,j-1})\sim\prod
^n_{i=1}p(x_{ri}|w_{1i},w_{2i})$.
\subsubsection*{Decoding at the relay}
At the end of block $j$, given $(m_{10,j-1}, m_{20,j-1})$ is known, the relay uses joint decoding to find the unique
$(\tilde{m}_{10,j},\tilde{m}_{20,j})$ such that
\begin{align}
  \label{dec_relay}
  (w_1^n,w_2^n,x_r^n,u_1^n(\tilde{m}_{10,j}|\cdot),
  u_2^n(\tilde{m}_{20,j}|\cdot),
  y^n_{r,j}) \in T^{(n)}_{\epsilon},
\end{align}
where $\cdot$ denotes known messages and $T^{(n)}_{\epsilon}$ denotes
the typical set. As in a multiple access channel, the relay can
success with vanishing error as $n\to\infty$ if
\begin{align}\label{sd1}
  R_{10} &\leq I(U_1 ;Y_r| W_1,W_2,U_2,U_r,X_r)\triangleq I_1; \nonumber\\
  R_{20} &\leq I(U_2 ;Y_r| W_1,W_2,U_1,U_r,X_r)\triangleq I_2; \nonumber\\
  R_{10}+R_{20} &\leq I(U_1,U_2 ;Y_r| W_1,W_2,U_r,X_r)\triangleq I_3.
\end{align}
\subsubsection*{Decoding at the user}
Each user uses forward sliding window decoding over 2 consecutive
blocks. Take user 2 for example. Given that it has correctly decoded
messages in all blocks up to $j-2$, at the end of block $j$, user 2
finds the unique $(\hat{m}_{10,j-1},\hat{m}_{11,j-1})$ such that
\begin{align*}
  \big(u_1^n(\hat{m}_{10,j-1}|\cdot),
  x_1^n(\hat{m}_{11,j-1}|\hat{m}_{10,j-1},\cdot),
  w_1^n,w_2^n,&\nonumber\\
  x_r^n,u_2^n,x_2^n, y^n_{2,j-1}\big)& \in T^{(n)}_{\epsilon} \nonumber \\
  \text{and} \;\;
  (w_1^n(\hat{m}_{10,j-1}),x_r^n(\hat{m}_{10,j-1},\cdot),
  w_2^n,u_2^n,  x_2^n,y^n_{2,b})& \in T^{(n)}_{\epsilon}.
\end{align*}
Following the standard error analysis in \cite{hsce612NIT}, we obtain the following rate constraints:
\begin{align}\label{sd2}
  R_{11} & \leq I(X_1;Y_2|W_1,U_1,W_2,U_2,X_2,X_r)\triangleq I_4, \nonumber\\
  R_{10}+R_{11} &\leq I(U_1,X_1;Y_2|W_1,W_2,U_2,X_2,X_r)\nonumber\\
  &\;\;  +I(W_1,X_r;Y_2| W_2,U_2, X_2)\triangleq I_{5}.
\end{align}
Similarly, for the decoding at user 1 to success with vanishing error,
\begin{align}\label{sd4}
R_{22} &\leq I(X_2;Y_1|W_1,U_1,W_2,U_2,X_1,X_r,U_r)\triangleq I_6\nonumber\\
R_{20}+R_{22} &\leq   I(U_2,X_2,W_2,X_r,U_r;Y_1|W_1,U_1,X_1)\triangleq I_{7}.
\end{align}
By performing Fourier-Mozkin elimination, we obtain the following
achievable rate region:
\begin{align}
  \label{symbol_constraint}
  R_1 &\leq \min\{I_1+I_4, I_{5}\},\nonumber\\
  R_2 &\leq \min\{I_2+I_{6}, I_{7}\},\nonumber\\
  R_1+R_2 &\leq I_3+I_4+I_{6},
\end{align}
for some joint distribution $P$ as in \eqref{code-distr} where $I_1$---$I_7$ are given in (\ref{sd1})---(\ref{sd4}).
Apply the signaling in \eqref{tx-signals} to Gaussian channel
\eqref{GTWRC}, we obtain the rate region in Theorem
\ref{thm:complete_pDF_gaussian}.
\setcounter{subsubsection}{0}
\section*{Appendix B: Proof of Theorem \ref{thm:4_cases}}\label{app:1}
For the independent pDF scheme, the achievable rate region as in Theorem 1 reduces to that in (\ref{eq:gaussian_par}).
We will analyze the region in (\ref{eq:gaussian_par}) here.
\begin{align}\label{eq:gaussian_par}
     &R_1 \leq \min \left\{J_1, J_4\right\},\;\;
    R_2 \leq \min \left\{J_2,J_5\right\},\nonumber\\
    &R_1+R_2 \leq J_3,\;\;\text{where} \\
    & J_1=C\left(\frac{g_{r1}^2{\beta_1}}{g_{r1}^2\gamma_1+g_{r2}^2\gamma_2+1}\right)+C(g_{21}^2\gamma_1),\nonumber\\
    &J_2=C\left(\frac{g_{r2}^2{\beta_2}}{g_{r1}^2\gamma_1+g_{r2}^2\gamma_2+1}\right)+C(g_{12}^2\gamma_2),\nonumber\\
    &J_3=C\left(\frac{g_{r1}^2{\beta_1}+g_{r2}^2{\beta_2}}{g_{r1}^2\gamma_1+g_{r2}^2\gamma_2+1}\right)+C(g_{21}^2\gamma_1)+C(g_{12}^2\gamma_2),
    \nonumber\\
    &J_4=C\left((g_{21}^2+g_{2r}^2)P\right),\;\;J_5=C\left((g_{12}^2+g_{1r}^2)P\right),\nonumber\\
    &0\leq {\beta_1}, {\beta_2} \leq 1,\;\;
    \beta_1 + \gamma_1=P, \;\;\beta_2 + \gamma_2=P,\nonumber\\\
    &\text{ and }\;\;
    C(x)=\log(1+x).\nonumber
  \end{align}
\subsection{Case $1$: Equations (\ref{case1}) and (\ref{case2})}
We only prove here the conditions in (\ref{case1:a}) as those in
(\ref{case1:b}) can be shown similarly. This case can be divided into the following $4$ subcases:
\subsubsection{Subcase $1.1$: $g_{21}^2<g_{r1}^2<\min\left\{g_{21}^2+g_{2r}^2\frac{P_r}{P_1},\;g_{21}^2(1+g_{r2}^2P_2)\right\}$}
\begin{figure}[!t]
    \begin{center}
    \includegraphics[width=0.4\textwidth, height=48mm]{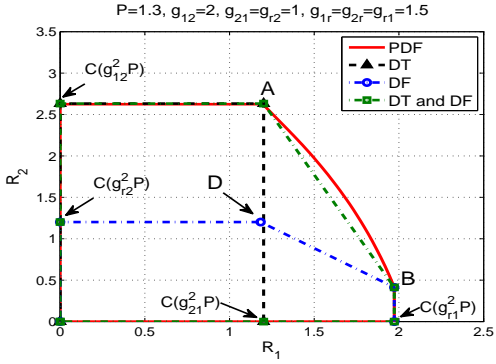}
    \caption{Subcase $1.1$.}
    \label{mcf}
    \end{center}
    \vspace{-6mm}
\end{figure}
\begin{figure}[!t]
    \begin{center}
    \includegraphics[width=0.4\textwidth, height=48mm]{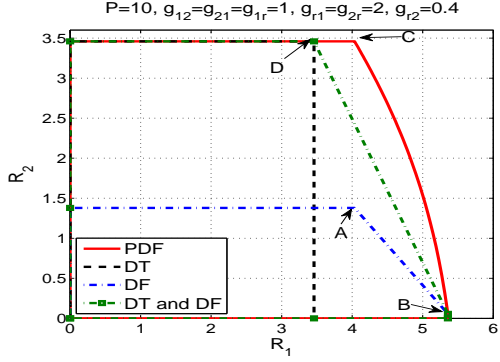}
    \caption{Subcase $1.2$.}\label{fig:app_case0}
    \end{center}
    \vspace{-6mm}
\end{figure}
\begin{figure}[t]
    \begin{center}
    \includegraphics[width=0.4\textwidth, height=48mm]{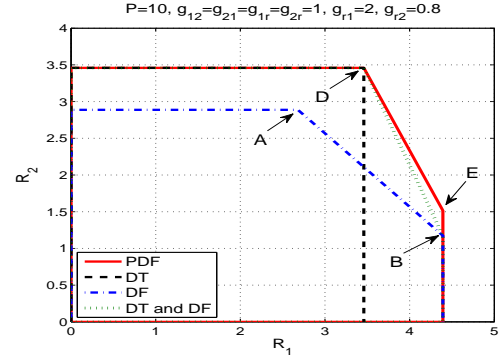}
    \caption{Subcase $1.3$.}\label{fig:app_case2}
    \end{center}
    \vspace{-6mm}
\end{figure}
\begin{figure}[!t]
    \begin{center}
    \includegraphics[width=0.4\textwidth, height=48mm]{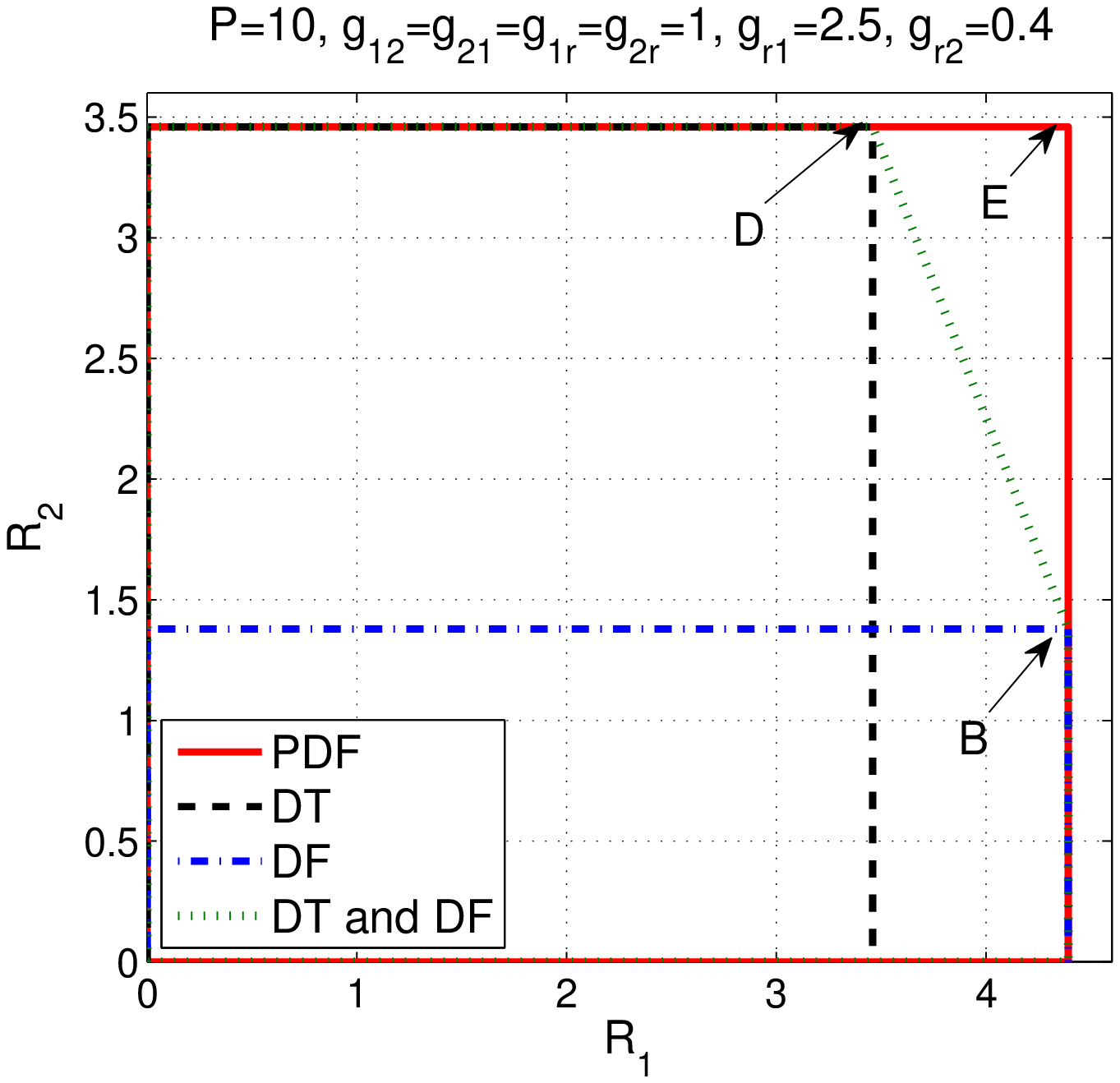}
    \caption{Case $2$.}\label{fig:app_case1}
    \end{center}
    \vspace{-6mm}
\end{figure}
We only prove here the conditions in (\ref{case2:b}) as (\ref{case2:e}) can be shown similarly.
Figure \ref{mcf} shows a typical achievable rate regions of the independent
pDF, DF and DT schemes for channel parameters that satisfy the conditions
in (\ref{case2:b}) but with $R_{S1}>R_{TS,1}$.

First, from the conditions $g_{12}^2>g_{r2}^2,$ $g_{r1}^2>g_{21}^2$ and
rate constraints of the DT and DF schemes, we can obtain the coordinates
for points $A$ and in $B$ in Figure \ref{mcf}.
Then we consider the convex hull of the DF and DT schemes, which is obtained
by connecting points $A$ and $B$ with a line that can be formulated as
\begin{align}\label{mqm3}
R_2+\mu_1 R_1&=R_{TS,1},\nonumber\\
\;\;\text{where}\;\mu_1&\;\text{and}\;R_{TS,1}\;\text{are given in (\ref{ffft})}.
\end{align}

Next, we consider the pDF scheme. It is easy to show that
$J_1<J_4$ since $g_{r1}^2>g_{21}^2$ and $g_{r1}^2<g_{21}^2+g_{2r}^2\frac{P_r}{P_1}$ while
$J_2<J_5$ since $g_{12}^2>g_{r2}^2$. Therefore, the rate region in
(\ref{eq:gaussian_par})  reduces to 
\begin{align}\label{independent pDF}
    R_1 &\leq J_1,\;\;
    R_2 \leq J_2,\;\;
    R_1+R_2 \leq J_3.
  \end{align}
We then maximize $R_2+\mu_1 R_1$ in (\ref{mqm3}) for the pDF scheme,
subject to rate constraints in \eqref{independent pDF},
and compare the results with $R_{TS,1}$.

Based the channel conditions for this subcase and by considering points $A$ and $D$ in Figure \ref{mcf}, we have $R_1(A)>R_1(D)$ and
$R_2(A)>R_2(D)$. Hence, the slope of the time sharing line (between $A$ and $B$) is greater than the slope  of the full DF scheme (between $D$ and $B$),
leading to $\mu>1$.
Furthermore,
the achievable rate for the pDF scheme in
(\ref{independent pDF}) is a pentagon for any fixed power allocation where the line connecting the upper and lower corner points has a slope of $1$.
Therefore, we only need to maximize the lower corner point (point $C$)
of this pentagon, i.e.
$(R_1(C),R_2(C))=(J_1,J_3-J_1)$. If this point is inside the time shared
region, the upper corner point will also be inside the time shared region.

By substituting $\beta_1=P_1-\gamma_1$ and $\beta_2=P_2-\gamma_2$ into the
lower corner point $(R_1(C),R_2(C))=(J_1,J_3-J_1)$,
we obtain the following problem to maximize $R_2+\mu_1 R_1$:
\begin{align}\label{maxdt}
\max_{\substack{\gamma_1, \gamma_2}} f(\gamma_1,\gamma_2)&,\quad \text{where}\\
 f(\gamma_1,\gamma_2)&=J_3-J_1+\mu_1 J_1,\nonumber\\
\Rightarrow f(\gamma_1,\gamma_2)&= C\left(g_{r1}^2P_1+g_{r2}^2P_2\right)-C\left(\frac{g_{r1}^2P_1}{1+g_{12}^2\gamma_2}\right)\nonumber\\
&\;-\mu_1 \left(C\left(\frac{g_{r1}^2{\beta_1}}{g_{r1}^2\gamma_1+g_{r2}^2\gamma_2+1}\right)+C(g_{21}^2\gamma_1)\right),\nonumber
\end{align}
and $J_1$ and $J_3$ are as given in (\ref{eq:gaussian_par}). We decompose
this problem into $2$ steps where we find the optimal $\gamma_1^{\ast}$ for a
fixed $\gamma_2$ and then find the optimal $\gamma_2^{\ast}$.
We find that $f(\gamma_1,\gamma_2)$ is either monotonically increasing or decreasing
in $\gamma_1$ depending on the value of $\gamma_2$ since
$\frac{\partial f(\gamma_1,\gamma_2)}{\partial \gamma_1}\gtrless 0\;\Leftrightarrow\;
\gamma_2\gtrless \frac{g_{r1}^2-g_{21}^2}{g_{21}^2g_{r2}^2}.$
Hence, for any value of $\gamma_2,$ the optimal $\gamma_1^{\ast}$ is either
$0$ or $P$. By setting $\gamma_1^{\ast}=P$, we find that
$\gamma_2^{\ast}=P$ since $\frac{\partial f(\gamma_1,\gamma_2)}
{\partial \gamma_2}>0$ under condition $g_{12}^2>g_{r2}^2$. At
$\gamma_1^{\ast}=\gamma_2^{\ast}=P$, we obtain the point $A$ in Figure
(\ref{mcf}) which is the same for both the pDF scheme and the DF-DT time shared
scheme. Hence, at this point, $R_2+\mu_1 R_1=R_{TS,1}$. On the other hand,
when $\gamma_1^{\ast}=0$, we obtain $\gamma_2^{\ast}$ in (\ref{ffft}) by setting $\frac{\partial f(\gamma_1,\gamma_2)}
{\partial \gamma_2}=0$. If $\gamma_2^{\ast}$ in  (\ref{ffft}) is negative, then $\frac{\partial f(\gamma_1,\gamma_2)}
{\partial \gamma_2}<0$ for $\gamma_2\in[0,P]$ and $f(0,\gamma_2)$ is a decreasing function of $\gamma_2$ which makes $\gamma_2^{\ast}=0$. Therefore, $\gamma_2^{\ast}$ is given as in (\ref{ffft}). After that, we have $R_2+\mu_1 R_1=f(0,\gamma_2^{\ast})$. Hence, to
compare with the time-shared scheme, we only need to compare the point at
$(\gamma_1,\gamma_2)=(0,\gamma_2^{\ast})$.

Specifically, we compare $f(0,\gamma_2^{\ast})$ with $R_{TS,1}$. If
$f(0,\gamma_2^{\ast})>R_{TS,1}$, then pDF achieves a larger rate region than
the DF-DT time shared region, which corresponds to (\ref{case12:c}) where
$f(0,\gamma_2^{\ast})=R_{S,1}$ in (\ref{ffft}). If $f(0,\gamma_2^{\ast})<R_{TS,1}$, then pDF
achieves the same region  as the DF-DT time shared region (condition
(\ref{case2:b})). Last, the curve connecting the points $A$ and $B$ is
obtained by varying $\gamma_2\in[0,P]$ while keeping
$\gamma_1^\star=0$. Hence the optimal scheme always have user 1 performing
full DF, while user 2 performing partial DF.
\subsubsection{Subcase $1.2$: $g_{21}^2(1+g_{r2}^2P_2)<g_{r1}^2<g_{21}^2+g_{2r}^2\frac{P_2}{P_1}$}
Figure \ref{fig:app_case0} illustrates a typical  example for this subcase. It can be proved by following a similar procedure to case $1.1$ except that we need to consider the time shared region of  the DF  scheme and the hybrid scheme (user 1: DF, user $2$: DT). The convex hull of this time shared region is
obtained by connecting points C and B with a line as in case $1.1,$ where the coordinates of point B is
$(R_1(B),R_2(B))=(C(g_{r1}^2P_1),C(\frac{g_{r2}^2P_2}{1+g_{r1}^2P_1}))$  and for point C is
$(R_1(C),R_2(C))=(C(\frac{g_{r1}^2P_1}{1+g_{r2}^2P_2}),C(g_{12}^2P_2))$.
%
\subsubsection{Subcase $1.3$: $g_{21}^2+g_{2r}^2\frac{P_r}{P_1}<g_{r1}^2<g_{21}^2(1+g_{r2}^2P_2)$}
%
%
Figure \ref{fig:app_case2}   illustrates a typical example for this case.
We show here that pDF can achieve point $E$ which is outside the DF-DT
time-shared region.
Staring with the full DF scheme $(\gamma_1=\gamma_2=0)$, the
rate region in (\ref{eq:gaussian_par}) becomes
\begin{align}
R_1&\overset{(a)}{\leq} C(g_{21}^2P_1+g_{2r}^2P_r),\;\;
R_2\overset{(b)}{\leq} C(g_{r2}^2P_2),\nonumber\\
R_1+R_2 &\leq C(g_{r1}^2P_1+g_{r2}^2P_2),
\end{align}
where $(a)$ follows since $g_{r1}^2>g_{21}^2+g_{2r}^2\frac{P_r}{P_1}$ while $(b)$ follows
since $g_{12}^2>g_{r2}^2$. Then, the coordinates of point $B$ in Figure
\ref{fig:app_case2} are $R_1(B)=C(g_{21}^2P_1+g_{2r}^2P_r)$ and
$R_2(B)=C(g_{r1}^2P_1+g_{r2}^2P_2)-C(g_{21}^2P_1+g_{2r}^2P_r)$.

Now consider the pDF scheme and its achievable region in
(\ref{eq:gaussian_par}). At $R_2=0$ $(\gamma_2=\beta_2=0)$,
$R_1^{\max}=C(g_{21}^2P_1+g_{2r}^2P_r)$ since $J_1 (\gamma_1)$ in
(\ref{eq:gaussian_par}) is a decreasing function of $\gamma_1$ and $J_1
(0)=C(g_{r1}^2P_1)>C(g_{21}^2P_1+g_{2r}^2P_r)$. Therefore, $R_1$ is the same for
points $B$ and $E$ $(R_1(E)=R_1(B))$.

Next, we need to find $R_2^{\max}$ such that
$R_1=C(g_{21}^2P_1+g_{2r}^2P_r)$. More specifically, we need to find
$\gamma_2^{\max}$ such that $J_1\geq J_4$. Consider $\gamma_1=0$, we have
\begin{align}
&J_1\geq J_4\Leftrightarrow
\log\left(1+\frac{g_{r1}^2P_1}{1+g_{r2}^2\gamma_2}\right)\geq
\log(1+g_{21}^2P_1+g_{2r}^2P_r)\nonumber\\
&\Leftrightarrow \gamma_2\leq
\frac{1}{g_{r2}^2}\left(\frac{g_{r1}^2}{g_{21}^2+g_{2r}^2\frac{P_r}{P_1}}-1\right)=\gamma_2^{\max}.\nonumber
\end{align}
With this value of $\gamma_2$,  $J_1(\gamma_1)$ is a decreasing function of
$\gamma_1$ since $\frac{\partial J_1}{\partial \gamma_1}<0$. Hence, setting
$\gamma_1=0$ is in fact optimal. Note that since $\gamma_2\leq P_2$, $\gamma_2^{\max}$ is given as follows:
\begin{align}\label{fmm}
&\gamma_2^{\max}=\\
&\left\{\begin{array}{cl}
 \frac{1}{g_{r2}^2}\left(\frac{g_{r1}^2}{g_{21}^2+g_{2r}^2\frac{P_r}{P_1}}-1\right)& \text{if}\;\; g_{r1}^2< \left(g_{21}^2+g_{2r}^2\frac{P_r}{P_1}\right)(1+g_{r2}^2P_2) \\
	P_2, &\text{if}\;\;g_{r1}^2\geq \left(g_{21}^2+g_{2r}^2\frac{P_r}{P_1}\right)(1+g_{r2}^2P_2)
\end{array}\right.\nonumber
\end{align}
However, since $g_{r1}^2<g_{21}^2(1+g_{r2}^2P_2)$ for this case, we have $\gamma_2^{\max}=\frac{1}{g_{r2}^2}\left(\frac{g_{r1}^2}{g_{21}^2+g_{2r}^2\frac{P_r}{P_1}}-1\right)$.
After that, define $J_i(\gamma_1,\gamma_2)$ for $i\in \{1:5\}$ as $J_i$ in
(\ref{eq:gaussian_par}) obtained with $\gamma_1$ and $\gamma_2$. Then, with
$\gamma_1=0$ and $\gamma_2=\gamma_2^{\max}$, the achievable rate region in
(\ref{eq:gaussian_par}) is given as
\begin{align}\label{rm3}
R_1&\leq C(g_{21}^2P_1+g_{2r}^2P_r),\;\;R_2\leq J_2(0,\gamma_2^{\max}),\nonumber\\
R_1+R_2&\leq J_3(0,\gamma_2^{\max}).
\end{align}
Then, $J_3(0,\gamma_2^{\max})<C(g_{21}^2P_1+g_{2r}^2P_r)+J_2(0,\gamma_2^{\max})$ since
\begin{align}
&J_3(0,\gamma_2^{\max})-J_2(0,\gamma_2^{\max})\nonumber\\
&        =C\left(\frac{g_{r1}^2P_1}{g_{r2}^2P_2+1}\right)
      \overset{(a)}{<}C(g_{21}^2P_1+g_{2r}^2P_r).\nonumber
\end{align}
where $(a)$ follows from (\ref{fmm}). Hence, in Figure \ref{fig:app_case2},
$R_2(E)=J_3(0,\gamma_2^{\max})-C(g_{21}^2P_1+g_{2r}^2P_r)>R_2(B)$ since
$R_2(E)-R_2(B)=C(g_{12}^2\gamma_2^{\max})-C(g_{r2}^2\gamma_2^{\max})>0.$
Then, point $E$ is outside the DF-DT time-shared region. To show switching between DF and DT is optimal for user $1$,
follow the same proof in subcase $1.1$.
\subsubsection{Subcase $1.4$: $\max\left\{g_{21}^2+g_{2r}^2\frac{P_r}{P_1},\;g_{21}^2(1+g_{r2}^2P_2)\right\}<g_{r1}^2<\left(g_{21}^2+g_{2r}^2\frac{P_r}{P_1}\right)(1+g_{r2}^2P_2)$}
This case includes both subcases $1.2$ and $1.3$ and has a similar proof to these cases.
\vspace*{-3mm}
\subsection{Case $2$: conditions (\ref{case12})}
This case can be simply proved by considering $\gamma_2^{\max}$ in (\ref{fmm}) for subcase $1.3$.
If $g_{r1}^2\geq \left(g_{21}^2+g_{2r}^2\frac{P_r}{P_1}\right)(1+g_{r2}^2P_2)$, then $\gamma_2=P_2$ and $\gamma_1=0$ which locates point $E$ in the upper corner as shown in Figure \ref{fig:app_case1}. Then,  point $E$ is also outside the DF-DT time-shared region.
Point $E$ and $C$ collide and they are achieved by user $1$ doing full DF and user $2$ doing DT. Since the rate region for this case is a rectangle as shown
in Figure \ref{fig:app_case1}, DF from user $1$ and DT from user $2$ are optimal to achieve the full rate region.
%
\vspace*{-3mm}
\subsection{Case $3$: conditions (\ref{case5})}
The conditions for this case are given in (\ref{case5}). Note that the
third condition is equivalent to  $C(g_{r1}^2P_1+g_{r2}^2P_2) < C(g_{21}^2P_1)+C(g_{12}^2P_2)$.
Considering the rate region in (\ref{eq:gaussian_par}), it easy to show that setting
$\gamma_1^{\ast}=\gamma_2^{\ast}=0$ maximizes the individual rates of this
region while setting $\gamma_1^{\ast}=P_1$ and $\gamma_2^{\ast}=P_2$ maximizes the sum
rate of this region. By taking the convex closure of the two regions
resulting from these two settings, we obtain the time shared region of the
DF and DT schemes.
\vspace*{-3mm}
\subsection{Case 4 and 5: conditions (\ref{case3}) and (\ref{case4})}
Considering rate region (\ref{eq:gaussian_par}) with the conditions
in (\ref{case3}) (reps. (\ref{case4})), we maximize the rate region by setting
$\gamma_1^{\ast}=\gamma_2^{\ast}=0$ (resp. $\gamma_1^{\ast}=P_1,\;\gamma_2^{\ast}=P_2$). Hence, the pDF scheme reduces to DF (resp. DT)
scheme.
%
%
%
\vspace*{-4mm}
\bibliographystyle{IEEEtran}
\bibliography{reflist}

\begin{thebibliography}{10}
\providecommand{\url}[1]{#1}
\csname url@samestyle\endcsname
\providecommand{\newblock}{\relax}
\providecommand{\bibinfo}[2]{#2}
\providecommand{\BIBentrySTDinterwordspacing}{\spaceskip=0pt\relax}
\providecommand{\BIBentryALTinterwordstretchfactor}{4}
\providecommand{\BIBentryALTinterwordspacing}{\spaceskip=\fontdimen2\font plus
\BIBentryALTinterwordstretchfactor\fontdimen3\font minus
  \fontdimen4\font\relax}
\providecommand{\BIBforeignlanguage}[2]{{%
\expandafter\ifx\csname l@#1\endcsname\relax
\typeout{** WARNING: IEEEtran.bst: No hyphenation pattern has been}%
\typeout{** loaded for the language `#1'. Using the pattern for}%
\typeout{** the default language instead.}%
\else
\language=\csname l@#1\endcsname
\fi
#2}}
\providecommand{\BIBdecl}{\relax}
\BIBdecl

\bibitem{DTD}
L.~Lei, Z.~Zhnog, C.~Lin, and X.~Shen, ``{Operator controlled device-to-device
  communications in LTE-advanced networks},'' \emph{IEEE Wireless Commun.},
  vol.~19, no.~3, pp. 96--104, Jun. 2012.

\bibitem{rankov2006achievable}
B.~Rankov and A.~Wittneben, ``Achievable rate regions for the two-way relay
  channel,'' in \emph{IEEE Int'l Symp. on Info. Theory (ISIT)}, 2006, pp.
  1668--1672.

\bibitem{xie2007network}
L.~Xie, ``Network coding and random binning for multi-user channels,'' in
  \emph{IEEE 10th Canadian Workshop on Info. Theory (CWIT)}, pp. 85--88.

\bibitem{PV_DC}
P.~Zhong and M.~Vu, ``Decode-forward and compute-forward coding schemes for the
  two-way relay channel,'' in \emph{IEEE Information Theory Workshop (ITW)},
  Oct. 2011, pp. 115--119.

\bibitem{kim2008per}
S.~Kim, P.~Mitran, and V.~Tarokh, ``Performance bounds for bidirectional coded
  cooperation protocols,'' \emph{IEEE Trans. on Info. Theory}, vol.~54, no.~11,
  pp. 5235--5241, Nov. 2008.

\bibitem{zhong2012partial}
P.~Zhong and M.~Vu, ``Partial decode-forward coding schemes for the gaussian
  two-way relay channel,'' in \emph{IEEE Int'l Conference on Commu. (ICC)},
  June 2012, pp. 2451--2456.

\bibitem{ASHARA}
K.~Ishaque~Ashar, V.~Prathyusha, S.~Bhashyam, and A.~Thangaraj, ``Outer bounds
  for the capacity region of a gaussian two-way relay channel,'' in \emph{50th
  Annual Allerton Conf. on Comm., Control, and Computing}, Oct. 2012, pp.
  1645--1652.

\bibitem{nw6p}
C.~Gong, G.~Yue, and X.~Wang, ``A transmission protocol for a cognitive
  bidirectional shared relay system,'' \emph{IEEE J. Sel. Top. Sign. Proces.},
  vol.~5, no.~1, pp. 160--170, Feb. 2011.

\bibitem{khafagyicc}
M.~Khafagy, A.~El-Keyi, M.~Nafie, and T.~ElBatt, ``Degrees of freedom for
  separated and non-separated half-duplex cellular mimo two-way relay
  channels,'' in \emph{IEEE Int'l Conference on Commu. (ICC)}, Jun. 2012, pp.
  2445--2450.

\bibitem{gerdes}
L.~Gerdes, M.~Riemensberger, and W.~Utschick, ``Bounds on the capacity regions
  of half-duplex gaussian mimo relay channels,'' \emph{EURASIP Journal on
  Advances in Signal Processing, Special Issue on Advanced Distributed Wireless
  Communication Techniques Theory and Practice}, 2013.

\bibitem{cover1979capacity}
T.~Cover and A.~El~Gamal, ``Capacity theorems for the relay channel,''
  \emph{IEEE Trans. on Info. Theory}, vol.~25, no.~5, pp. 572--584, 1979.

\bibitem{CST1}
C.~Schnurr, S.~Stanczak, and T.~Oechtering, ``Achievable rates for the
  restricted half-duplex two-way relay channel under a
  partial-decode-and-forward protocol,'' in \emph{IEEE Information Theory
  Workshop (ITW)}, May 2008, pp. 134--138.

\bibitem{gmze}
A.~El~Gamal, M.~Mohseni, and S.~Zahedi, ``Bounds on capacity and minimum
  energy-per-bit for \uppercase{AWGN} relay channels,'' \emph{IEEE Trans. on
  Info. Theory}, vol.~52, no.~4, pp. 1545--1561, Apr. 2006.

\bibitem{chenicc}
Z.~Chen, H.~Liu, and W.~Wang, ``On the optimization of decode-andforward
  schemes for two-way asymmetric relaying,'' in \emph{IEEE Int'l Conference on
  Commu. (ICC)}, Jun. 2011, pp. 1--5.

\bibitem{op2p}
Y.~Shim and H.~Park, ``A closed-form expression of optimal time for two-way
  relay using \uppercase{DF MABC} protocol,'' \emph{IEEE Communications
  Letters}, vol.~PP, no.~99, pp. 1--4, 2014.

\bibitem{ROFD1}
K.~Jitvanichphaibool, R.~Zhang, and Y.-C. Liang, ``Optimal resource allocation
  for two-way relay-assisted \uppercase{OFDMA},'' \emph{IEEE Trans. Veh.
  Technol.}, vol.~58, no.~7, pp. 3311--3321, Sept. 2009.

\bibitem{ROFD2}
F.~He, Y.~Sun, L.~Xiao, X.~Chen, C.-Y. Chi, and S.~Zhou, ``Optimal resource
  allocation for two-way relay-assisted \uppercase{OFDMA},'' \emph{IEEE Trans.
  Wireless Commun.}, vol.~12, no.~6, pp. 2904--2917, Jun. 2013.

\bibitem{lisa2015}
L.~Pinals and M.~Vu, ``Link-state optimized decode-forward transmission for
  two-way relaying,'' \emph{submitted to IEEE Trans. on Commu.}, Jan. 2015.

\bibitem{gamal2010lecture}
A.~El~Gamal and Y.-H. Kim, \emph{Network Information Theory}.\hskip 1em plus
  0.5em minus 0.4em\relax Cambridge University Press, 2011.

\bibitem{lisa2014}
L.~Pinals and M.~Vu, ``Link state based decode-forward schemes for two-way
  relaying,'' in \emph{International Workshop on Emerging Technologies for 5G
  Wireless Cellular Networks, IEEE Globecom}, Dec. 2014.

\bibitem{raport}
T.~S. Rappaport, \emph{{Wireless Communications: Principles and Practice}},
  2nd~ed.\hskip 1em plus 0.5em minus 0.4em\relax Prentice Hall, 2002.

\bibitem{hsce612NIT}
R.~El~Gamal and Y.-H. Kim, \emph{{Network Information Theory}}, 1st~ed.\hskip
  1em plus 0.5em minus 0.4em\relax Cambridge University Press, 2011.

\end{thebibliography}

\end{document}